\documentclass[letterpaper,11pt]{article}

\usepackage[T1]{fontenc}
\usepackage[letterpaper,top=1in,bottom=1in,left=1in,right=1in]{geometry}
\usepackage{amsmath,amsthm}
\usepackage{amssymb,latexsym}
\usepackage[mathscr]{eucal}
\usepackage{dsfont}
\usepackage{graphicx,color}
\usepackage{psfrag}
\usepackage{xspace}
\usepackage{comment}
\usepackage[unicode,final=true]{hyperref} 

\usepackage{multirow,bigdelim}

\newtheorem{theorem}{Theorem}

\newtheorem{lemma}[theorem]{Lemma}

\newtheorem{claim}[theorem]{Claim}

\newcounter{pln}

\newcommand{\reals}{\mathbb{R}}
\newcommand{\naturals}{\mathbb{N}}

\newcommand{\cC}{\mathcal{C}}
\newcommand{\wW}{\mathcal{W}}
\newcommand{\pP}{\mathcal{P}}
\newcommand{\qQ}{\mathcal{Q}}

\newcommand{\Wavg}{\overline{W}}
\newcommand{\mblam}{\mb{\lambda}}
\newcommand{\lambdaavg}{\overline{\mblam}}

\newcommand{\UU}{\mathscr{U}}
\newcommand{\VV}{\mathscr{V}}
\newcommand{\UUC}{\mathscr{U}'}
\newcommand{\VVC}{\mathscr{V}'}
\newcommand{\SSS}{\mathscr{S}}

\renewcommand{\to}{\rightarrow}
\newcommand{\sub}{\leftarrow}

\newcommand{\conv}{{\rm conv}}
\newcommand{\sm}{\setminus}
\newcommand{\nin}{\not\in}

\newcommand{\mb}[1]{{\boldsymbol{#1}}}
\newcommand{\me}{\mb{e}}
\newcommand{\mbv}{\mb{v}}
\newcommand{\mbu}{\mb{u}}
\newcommand{\mc}{\mb{\gamma}}
\newcommand{\mcc}{\mb{\beta}}

\newcommand{\sep}{,\xspace}
\newcommand{\resp}{resp.,\xspace}
\newcommand{\ie}{i.e.,\xspace}
\newcommand{\eg}{e.g.,\xspace}

\newcommand{\lexp}[1]{\lfloor\frac{#1}{2}\rfloor}
\newcommand{\ltexp}[1]{\lfloor\tfrac{#1}{2}\rfloor}
\newcommand{\dup}{d+r-1}

\newcommand{\ddo}{d-r+1}
\newcommand{\pddo}{(d-r+1)}
\newcommand{\bx}[1]{\partial{#1}}

\newcommand{\ptn}[1][V]{\mathsf{#1}}
\newcommand{\ptnc}[1][V]{\mathsf{#1}'}

\newcommand{\spansub}[1][\ptn]{${#1}$-spanning subset\xspace}
\newcommand{\spansubs}[1][\ptn]{${#1}$-spanning subsets\xspace}

\newcommand{\range}[2][r]{{#2}_1,{#2}_2,\ldots,{#2}_{#1}}

\newcommand{\ub}[1]{\Phi_{#1}(\range{n})}
\newcommand{\MSsymbol}{+}
\newcommand{\MS}[1][P]{{#1}_1\MSsymbol{#1}_2\MSsymbol\cdots\MSsymbol{#1}_r}
\newcommand{\WMS}[1][P]{\lambda_1{#1}_1\MSsymbol\lambda_2{#1}_2\MSsymbol\cdots\MSsymbol\lambda_{r}{#1}_r}

\newcommand{\csep}{\,\,\,\,\,\,}

\newcommand{\compl}{\bar{S}}

\newcommand{\mddet}{\Delta_{(\kappa_1,\ldots,\kappa_n)}(\tau)}
\newcommand{\md}{\mb{\Delta}_{(\kappa_1,\ldots,\kappa_n)}(\tau)}
\newcommand{\VD}[2][]{\text{VD}(#1\mb{#2})}
\newcommand{\GVD}{\text{GVD}}

\newcommand{\hide}[1]{}

\newcounter{ctrcopy}
\newcommand{\savecounter}[1]{\setcounter{ctrcopy}{#1}}
\newcommand{\restorecounter}[1]{%
  \setcounter{#1}{\thectrcopy}%
  \addtocounter{#1}{-1}}

\begin{document}

% title
\title{Tight lower bounds on the number of faces of the Minkowski sum of
  convex polytopes via the Cayley trick}

% author/paper info
\author{Menelaos I. Karavelas$^{\dagger,\ddagger}$\hfil{}
Eleni Tzanaki$^{\dagger,\ddagger}$\\[5pt]
\it{}$^\dagger$Department of Applied Mathematics,\\
\it{}University of Crete\\
\it{}GR-714 09 Heraklion, Greece\\
{\small\texttt{\{mkaravel,etzanaki\}@tem.uoc.gr}}\\[5pt]
\it{}$^\ddagger$Institute of Applied and Computational Mathematics,\\
\it{}Foundation for Research and Technology - Hellas,\\
\it{}P.O. Box 1385, GR-711 10 Heraklion, Greece}

%\author{Menelaos I. Karavelas}
%\address{Department of Applied Mathematics, University of Crete,
%  GR-714 09 Heraklion, Greece}
%\curraddr{lalaMK}
%\email{mkaravel@tem.uoc.gr}
%\thanks{thanksMK}

%\author{Eleni Tzanaki}
%\address{Department of Applied Mathematics, University of Crete,
%  GR-714 09 Heraklion, Greece}
%\curraddr{lalaET}
%\email{etzanaki@tem.uoc.gr}
%\thanks{thanksET}

\phantomsection
\addcontentsline{toc}{section}{Title}

\maketitle

\phantomsection
\addcontentsline{toc}{section}{Abstract}

\begin{abstract}
  Consider a set of $r$ convex $d$-polytopes $P_1,P_2,\ldots,P_r$, where
$d\ge{}3$ and $r\ge{}2$, and let $n_i$ be the number of vertices of
$P_i$, $1\le{}i\le{}r$.
It has been shown by Fukuda and Weibel \cite{fw-fmacp-07} that the
number of $k$-faces of the Minkowski sum, $P_1+P_2+\cdots+P_r$, is
bounded from above by $\Phi_{k+r}(n_1,n_2,\ldots,n_r)$, where
$\Phi_{\ell}(n_1,n_2,\ldots,n_r)=
\sum_{\substack{1\le{}s_i\le{}n_i\\s_1+\ldots+s_r=\ell}}
\prod_{i=1}^r\binom{n_i}{s_i}$, $\ell\ge{}r$.
Fukuda and Weibel \cite{fw-fmacp-07}
have also shown that the upper bound mentioned above is tight for
$d\ge{}4$, $2\le{}r\le{}\lfloor\frac{d}{2}\rfloor$, and for all
$0\le{}k\le{}\lfloor\frac{d}{2}\rfloor-r$.

In this paper we construct a set of $r$ neighborly
$d$-polytopes $P_1,P_2,\ldots,P_r$, where $d\ge{}3$ and
$2\le{}r\le{}d-1$, for which the upper bound of Fukuda and Weibel is
attained for all $0\le{}k\le{}\lfloor\frac{d+r-1}{2}\rfloor-r$. A
direct consequence of our result is a tight asymptotic bound on the
complexity of the Minkowski sum $P_1+P_2+\cdots+P_r$, for any fixed
dimension $d$ and any $2\le{}r\le{}d-1$, when the number of vertices
of the polytopes is (asymptotically) the same.

Our approach is based on what is known as the Cayley trick for
Minkowski sums: the Minkowski sum, $P_1+P_2+\ldots+P_r$, is the
intersection of the Cayley polytope $\mathcal{P}$, in
$\mathbb{R}^{r-1}\times\mathbb{R}^{d}$, of the $d$-polytopes
$P_1,P_2,\ldots,P_r$, with an appropriately defined $d$-flat
$\overline{W}$ of $\mathbb{R}^{r-1}\times\mathbb{R}^{d}$. 
To prove our bounds, we construct the $d$-polytopes
$P_1,P_2,\ldots,P_r$, where $d\ge{}3$ and $2\le{}r\le{}d-1$, in such a
way so that the number of $(k-1)$-faces of $\mathcal{P}$, that intersect
the $d$-flat $\overline{W}$, is equal to
$\Phi_{k}(n_1,n_2,\ldots,n_r)$, for all
$r\le{}k\le\lfloor\frac{d+r-1}{2}\rfloor$. The tightness of our bounds
then follows from the Cayley trick: the $(k+r-1)$-faces of the intersection
of $\mathcal{P}$ with $\overline{W}$
are in one-to-one correspondence with the $k$-faces of
$P_1+P_2+\cdots+P_r$, which implies that
$f_k(P_1+P_2+\cdots+P_r)=\Phi_{k+r}(n_1,n_2,\ldots,n_r)$, for all
$0\le{}k\le{}\lfloor\frac{d+r-1}{2}\rfloor-r$.

  \bigskip\noindent
  \textit{Key\;words:}\/
  high-dimensional geometry\sep discrete geometry\sep
  combinatorial geometry\sep Cayley trick\sep lower bounds\sep
  Minkowski sum\sep convex polytopes
  \medskip\par\noindent
\end{abstract}

\clearpage

\newcommand{\techlemma}[1]{
\begin{lemma}{#1}
  Let $\kappa_1,\ldots,\kappa_n$ be $n\ge{}2$ integers such that
  $\kappa_i\ge{}2$, $1\le{}i\le{}n$, and let $K=\sum_{i=1}^{n}\kappa_i$. Let
  $x_{i,j}$ be real numbers such that
  $0<x_{i,1}<x_{i,2}<\ldots<x_{i,\kappa_i}$, $1\le{}i\le{}n$. Let
  $\beta_i$, $1\le{}i\le{}n$ be non-negative integers such that
  $\beta_1>\beta_2>\ldots>\beta_n\ge{}0$. Finally, let $\tau$ be a positive real
  parameter, and define $\mddet$ to be the determinant:
  \begin{center}
    \includegraphics[width=\textwidth]{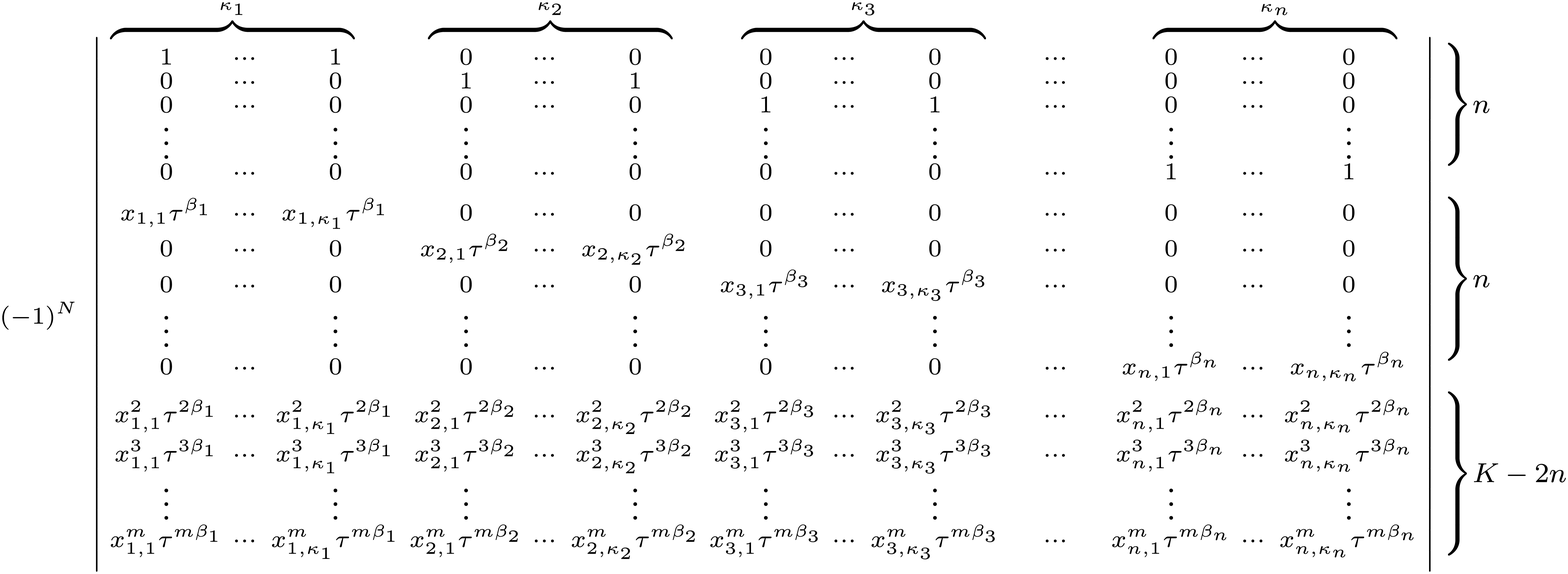}
  \end{center}
  where $N=\frac{n(n-1)}{2}$ and $m=K-2n+1$. 
  Then, there exists some $\tau_0>0$, such that for all
  $\tau\in(0,\tau_0)$,the determinant $\mddet$ is strictly positive.
\end{lemma}}

\newcommand{\Hexample}{%
  For example, for $d=8$, $r=3$, $k=4$, $|U_1|=|U_3|=1$, $|U_2|=2$,
  and $\nu_i=3-i$, $i=1,2,3$, $H_U(\mb{x})$ is the $11\times{}11$
  determinant:
  \begin{equation*}
    -\left|\begin{array}{@{\,}c@{\,\,\,}c@{\,\,\,}c@{\,\,\,}c@{\,\,\,}c@{\,\,\,}c@{\,\,\,}c@{\,\,\,}c@{\,\,\,}c@{\,\,\,}c@{\,\,\,}c@{\,}}%{@{\,}ccccccccccc@{\,}}
        1&1&1&1&1&1&1&1&1&1&1\\[3pt]
        x_1&0&0&1&1&1&1&0&0&0&0\\[3pt]
        x_2&0&0&0&0&0&0&1&1&1&1\\[3pt]
        x_3&\alpha_{1,j_{1,1}}\tau^2&\alpha_{1,j_{1,1}}^\epsilon\tau^2
        &0&0&0&0&0&0&0&0\\[3pt]
        x_4&0&0&\alpha_{2,j_{2,1}}\tau&\alpha_{2,j_{2,1}}^\epsilon\tau
        &\alpha_{2,j_{2,2}}\tau&\alpha_{2,j_{2,2}}^\epsilon\tau&0&0&0&0\\[3pt]
        x_5&0&0&0&0&0&0&\alpha_{3,j_{3,1}}&\alpha_{3,j_{3,1}}^\epsilon&M&2M\\[3pt]
        x_6&\alpha_{1,j_{1,1}}^2\tau^4&(\alpha_{1,j_{1,1}}^\epsilon)^2\tau^4
        &\alpha_{2,j_{2,1}}^2\tau^2&(\alpha_{2,j_{2,1}}^\epsilon)^2\tau^2
        &\alpha_{2,j_{2,2}}^2\tau^2&(\alpha_{2,j_{2,2}}^\epsilon)^2\tau^2
        &\alpha_{3,j_{3,1}}^2&(\alpha_{3,j_{3,1}}^\epsilon)^2&M^2&(2M)^2\\[3pt]
        x_7&\alpha_{1,j_{1,1}}^3\tau^6&(\alpha_{1,j_{1,1}}^\epsilon)^3\tau^6
        &\alpha_{2,j_{2,1}}^3\tau^3&(\alpha_{2,j_{2,1}}^\epsilon)^3\tau^3
        &\alpha_{2,j_{2,2}}^3\tau^3&(\alpha_{2,j_{2,2}}^\epsilon)^3\tau^3
        &\alpha_{3,j_{3,1}}^3&(\alpha_{3,j_{3,1}}^\epsilon)^2&M^3&(2M)^3\\[3pt]
        x_8&\alpha_{1,j_{1,1}}^4\tau^8&(\alpha_{1,j_{1,1}}^\epsilon)^4\tau^8
        &\alpha_{2,j_{2,1}}^4\tau^4&(\alpha_{2,j_{2,1}}^\epsilon)^4\tau^4
        &\alpha_{2,j_{2,2}}^4\tau^4&(\alpha_{2,j_{2,2}}^\epsilon)^4\tau^4
        &\alpha_{3,j_{3,1}}^4&(\alpha_{3,j_{3,1}}^\epsilon)^2&M^4&(2M)^4\\[3pt]
        x_9&\alpha_{1,j_{1,1}}^5\tau^{10}&(\alpha_{1,j_{1,1}}^\epsilon)^5\tau^{10}
        &\alpha_{2,j_{2,1}}^5\tau^5&(\alpha_{2,j_{2,1}}^\epsilon)^5\tau^5
        &\alpha_{2,j_{2,2}}^5\tau^5&(\alpha_{2,j_{2,2}}^\epsilon)^5\tau^5
        &\alpha_{3,j_{3,1}}^5&(\alpha_{3,j_{3,1}}^\epsilon)^2&M^5&(2M)^5\\[3pt]
        x_{10}&\alpha_{1,j_{1,1}}^6\tau^{12}&(\alpha_{1,j_{1,1}}^\epsilon)^6\tau^{12}
        &\alpha_{2,j_{2,1}}^6\tau^6&(\alpha_{2,j_{2,1}}^\epsilon)^6\tau^6
        &\alpha_{2,j_{2,2}}^6\tau^6&(\alpha_{2,j_{2,2}}^\epsilon)^6\tau^6
        &\alpha_{3,j_{3,1}}^6&(\alpha_{3,j_{3,1}}^\epsilon)^2&M^6&(2M)^6
      \end{array}\right|.
  \end{equation*}
}

%%%%%%%%%%%%%%%%%%%%%%%%%%%%%%%%%%%%%%%%%%%%%%%%%%%%%%%%%%%%%%%%%%%%%%%%%
%%%%%%%%%%%%%%%%%%%%%%%%%%%%%%%%%%%%%%%%%%%%%%%%%%%%%%%%%%%%%%%%%%%%%%%%%

\newcommand{\Fexample}{%
  For example, for $d=8$, $r=3$, $k=4$, $|V_1|=|V_3|=1$, $|V_2|=2$,
  and $\nu_i=3-i$, $i=1,2,3$, $F_V(\mb{x};\zeta)$ is the
  $11\times{}11$ determinant:
  \begin{equation*}
    -\left|\begin{array}{@{\,}c@{\,\,\,}c@{\,\,\,}c@{\,\,\,}c@{\,\,\,}c@{\,\,\,}c@{\,\,\,}c@{\,\,\,}c@{\,\,\,}c@{\,\,\,}c@{\,\,\,}c@{\,}}
        1&1&1&1&1&1&1&1&1&1&1\\[3pt]
        x_1&0&0&1&1&1&1&0&0&0&0\\[3pt]
        x_2&0&0&0&0&0&0&1&1&1&1\\[3pt]
        x_3&\alpha_{1,j_{1,1}}\tau^2&\alpha_{1,j_{1,1}}^\epsilon\tau^2
        &\zeta\alpha_{2,j_{2,1}}^7\tau^7&\zeta(\alpha_{2,j_{2,1}}^\epsilon)^7\tau^7
        &\zeta\alpha_{2,j_{2,2}}^7\tau^7&\zeta(\alpha_{2,j_{2,2}}^\epsilon)^7\tau^7
        &\zeta\alpha_{3,j_{3,1}}^7&\zeta(\alpha_{3,j_{3,1}}^\epsilon)^7
        &\zeta{}M^7&\zeta(2M)^7\\[3pt]
        x_4&\zeta\alpha_{1,j_{1,1}}^7\tau^{14}&\zeta(\alpha_{1,j_{1,1}}^\epsilon)^7\tau^{14}
        &\alpha_{2,j_{2,1}}\tau&\alpha_{2,j_{2,1}}^\epsilon\tau
        &\alpha_{2,j_{2,2}}\tau&\alpha_{2,j_{2,2}}^\epsilon\tau
        &\zeta\alpha_{3,j_{3,1}}^8&\zeta(\alpha_{3,j_{3,1}}^\epsilon)^8
        &\zeta{}M^8&\zeta(2M)^8\\[3pt]
        x_5&\zeta\alpha_{1,j_{1,1}}^8\tau^{16}&\zeta(\alpha_{1,j_{1,1}}^\epsilon)^8\tau^{16}
        &\zeta\alpha_{2,j_{2,1}}^8\tau^8&\zeta(\alpha_{2,j_{2,1}}^\epsilon)^8\tau^8
        &\zeta\alpha_{2,j_{2,2}}^8\tau^8&\zeta(\alpha_{2,j_{2,2}}^\epsilon)^8\tau^8
        &\alpha_{3,j_{3,1}}&\alpha_{3,j_{3,1}}^\epsilon&M&2M\\[3pt]
        x_6&\alpha_{1,j_{1,1}}^2\tau^4&(\alpha_{1,j_{1,1}}^\epsilon)^2\tau^4
        &\alpha_{2,j_{2,1}}^2\tau^2&(\alpha_{2,j_{2,1}}^\epsilon)^2\tau^2
        &\alpha_{2,j_{2,2}}^2\tau^2&(\alpha_{2,j_{2,2}}^\epsilon)^2\tau^2
        &\alpha_{3,j_{3,1}}^2&(\alpha_{3,j_{3,1}}^\epsilon)^2&M^2&(2M)^2\\[3pt]
        x_7&\alpha_{1,j_{1,1}}^3\tau^6&(\alpha_{1,j_{1,1}}^\epsilon)^3\tau^6
        &\alpha_{2,j_{2,1}}^3\tau^3&(\alpha_{2,j_{2,1}}^\epsilon)^3\tau^3
        &\alpha_{2,j_{2,2}}^3\tau^3&(\alpha_{2,j_{2,2}}^\epsilon)^3\tau^3
        &\alpha_{3,j_{3,1}}^3&(\alpha_{3,j_{3,1}}^\epsilon)^3&M^3&(2M)^3\\[3pt]
        x_8&\alpha_{1,j_{1,1}}^4\tau^8&(\alpha_{1,j_{1,1}}^\epsilon)^4\tau^8
        &\alpha_{2,j_{2,1}}^4\tau^4&(\alpha_{2,j_{2,1}}^\epsilon)^4\tau^4
        &\alpha_{2,j_{2,2}}^4\tau^4&(\alpha_{2,j_{2,2}}^\epsilon)^4\tau^4
        &\alpha_{3,j_{3,1}}^4&(\alpha_{3,j_{3,1}}^\epsilon)^4&M^4&(2M)^4\\[3pt]
        x_9&\alpha_{1,j_{1,1}}^5\tau^{10}&(\alpha_{1,j_{1,1}}^\epsilon)^5\tau^{10}
        &\alpha_{2,j_{2,1}}^5\tau^5&(\alpha_{2,j_{2,1}}^\epsilon)^5\tau^5
        &\alpha_{2,j_{2,2}}^5\tau^5&(\alpha_{2,j_{2,2}}^\epsilon)^5\tau^5
        &\alpha_{3,j_{3,1}}^5&(\alpha_{3,j_{3,1}}^\epsilon)^5&M^5&(2M)^5\\[3pt]
        x_{10}&\alpha_{1,j_{1,1}}^6\tau^{12}&(\alpha_{1,j_{1,1}}^\epsilon)^6\tau^{12}
        &\alpha_{2,j_{2,1}}^6\tau^6&(\alpha_{2,j_{2,1}}^\epsilon)^6\tau^6
        &\alpha_{2,j_{2,2}}^6\tau^6&(\alpha_{2,j_{2,2}}^\epsilon)^6\tau^6
        &\alpha_{3,j_{3,1}}^6&(\alpha_{3,j_{3,1}}^\epsilon)^6&M^6&(2M)^6
      \end{array}\right|.
  \end{equation*}
}

\section{Introduction}
\label{sec:intro}

Given two sets $A$ and $B$ in $\reals^{d}$, their Minkowski sum,
$A\MSsymbol{}B$, is defined as the set
$\{a+b\mid{}a\in{}A,b\in{}B\}$. Minkowski sums are fundamental
structures in both Mathematics and Computer Science. They appear in a
wide variety of sub-disciplines, including Combinatorial Geometry,
Computational Geometry, Computer Algebra, Computer-Aided Design and
Robotics, just to name a few. In recent years, they have found
applications in areas such as Game Theory and Computational
Biology. It is beyond the scope of this paper to discuss the different
uses and applications of Minkowski sums; the interested reader may
refer to \cite{w-mspcc-07} and \cite{f-mscaa-08},
and the references therein.

The focus of this work is on Minkowski sums of polytopes, and, in
particular, convex polytopes. Tight, or almost tight, asymptotic
bounds on the worst-case complexity of the Minkowski sum of two,
possibly non-convex, polytopes may be found, \eg in
\cite{bkos-cgaa-00}, \cite{s-amp-04}, \cite{fhw-emcms-09},
and \cite{kt-chsch-11b}.
In this paper, we are interested in \emph{exact} 
bounds on the complexity of the Minkowski sum of two or more
polytopes\footnote{In the rest of the paper all polytopes are
  considered to be convex.}. Our aim is to answer a natural and
fundamental question: given $r$ $d$-polytopes,
what is the maximum number of $k$-faces of their Minkowski sum?

For $r\ge{}2$ polygons (2-polytopes) $\range{P}$, it is
known that the number of vertices (or edges) of $\MS$ is equal, in
the worst-case, to $\sum_{i=1}^{r}n_i$, where $n_i$ is the number of
vertices (or edges) of $P_i$ (see, \eg \cite{bkos-cgaa-00},
\cite{w-mspcc-07}).
For higher-dimensional polytopes, the first answer to this question
was given by Gritzmann and Sturmfels \cite{gs-mapcc-93}: given $r$
polytopes $\range{P}$ in $\reals^d$, with a total of $n$ non-parallel edges,
the number of $l$-faces of $\MS$ is bounded from above by
$2\binom{n}{l}\sum_{j=1}^{d-1-l}\binom{n-l-1}{j}$. This bound is
attained if the polytopes $P_i$ are \emph{zonotopes}, whose generating
edges are in general position.
Fukuda and Weibel \cite{fw-fmacp-07} have shown, what they call the
\emph{trivial upper bound}: given $r$ $d$-polytopes
$\range{P}$, where $d\ge{}3$ and $r\ge{}2$, we have
\begin{equation}\label{equ:trivial-ub}
  f_k(\MS)\le{}\ub{k+r},
\end{equation}
where $n_i$ is the number of vertices of $P_i$, $1\le{}i\le{}r$, and
\begin{equation}\label{equ:phi-def}
  \ub{\ell}=\sum_{\substack{1\le{}s_i\le{}n_i\\s_1+\ldots+s_r=\ell}}
  \prod_{i=1}^r\binom{n_i}{s_i},\qquad \ell\ge{}r,\qquad s_i\in\naturals.
\end{equation}
%and the $s_i$'s are integral.
In the same paper, Fukuda and Weibel have shown that the trivial upper
bound is tight for: (i) $d\ge{}4$, $2\le{}r\le\lexp{d}$ and for all
$0\le{}k\le\lexp{d}-r$, and (ii) for the number of vertices,
$f_0(\MS)$, of $\MS$, when $d\ge{}3$ and $2\le{}r\le{}d-1$.
For $r\ge{}d$, Sanyal \cite{s-tovnm-09} has shown that the trivial
bound for $f_0(\MS)$ cannot be attained, since in this case:
\begin{equation*}
  f_0(\MS)\le\left(1-\frac{1}{(d+1)^d}\right)\prod_{i=1}^{r}n_{i}
  <\prod_{i=1}^{r}n_{i}.
\end{equation*}
Tight bounds for $f_0(\MS)$, where $r\ge{}d$, have very recently be
shown by Weibel \cite{w-mfmsl-11}, namely:
\begin{equation*}
  f_0(\MS)\le\alpha+\sum_{j=1}^{d-1}(-1)^{d-1-j}\binom{r-1-j}{d-1-j}
  \sum_{S\in\cC_j^r}\left(\prod_{i\in{}S}f_0(P_i)-\alpha\right),
\end{equation*}
where $\cC_j^r$ is the family of subsets of $\{1,2,\ldots,r\}$ of
cardinality $j$, and $\alpha=2(d-2\lexp{d})$.
%$\alpha=2$ if $d$ is odd, while $\alpha=0$, if $d$ is even.

Tight bounds for \emph{all} face numbers, \ie for all
$0\le{}k\le{}d-1$, are only known for \emph{two} $d$-polytopes, where
$d\ge{}3$.
Fukuda and Weibel \cite{fw-fmacp-07} have shown that, given two
3-polytopes $P_1$ and $P_2$ in $\reals^3$, the number of
$k$-faces of $P_1\MSsymbol{}P_2$, $0\le{}k\le{}2$, is bounded from
above as follows:
\begin{equation*}
  \begin{aligned}
    f_0(P_1\MSsymbol{}P_2)&\le{}n_1 n_2,\\
    f_1(P_1\MSsymbol{}P_2)&\le{}2n_1 n_2+n_1+n_2-8,\\
    f_2(P_1\MSsymbol{}P_2)&\le{}n_1 n_2+n_1+n_2-6,
  \end{aligned}
\end{equation*}
where $n_i$ is the number of vertices of $P_i$, $i=1,2$. These bounds
are tight.
Weibel \cite{w-mspcc-07} has derived analogous tight expressions in
terms of the number of facets $m_i$ of $P_i$, $i=1,2$:
\begin{equation}\label{equ:ub3}
  \begin{aligned}
    f_0(P_1\MSsymbol{}P_2)&\le{}4m_1 m_2-8m_1-8m_2+16,\\
    f_1(P_1\MSsymbol{}P_2)&\le{}8m_1 m_2-17m_1-17m_2+40,\\
    f_2(P_1\MSsymbol{}P_2)&\le{}4m_1 m_2-9m_1-9m_2+26.
  \end{aligned}
\end{equation}
Weibel's expression for $f_2(P_1\MSsymbol{}P_2)$
(cf. \eqref{equ:ub3}) has been generalized to the number of
facets of the Minkowski sum of any number of 3-polytopes
by Fogel, Halperin and Weibel \cite{fhw-emcms-09}; they have shown
that, for $r\ge{}2$, the following tight bound holds:
\begin{equation*}
  f_2(\MS)\le\sum_{1\le{}i<j\le{}r}(2m_i-5)(2m_j-5)
  +\sum_{i=1}^r{}m_i+\binom{r}{2},
\end{equation*}
where $m_i=f_2(P_i)$, $1\le{}i\le{}r$. %, and this bound is tight.
Regarding $d$-polytopes, where $d\ge{}4$, Karavelas and Tzanaki
\cite{kt-mnfms-11b,kt-mnfms-12}, have shown that, for 
$1\le{}k\le{}d$:
\begin{equation}\label{equ:2p-any-d}
  f_{k-1}(P_1\MSsymbol{}P_2)\le{}f_k(C_{d+1}(n_1+n_2))
  -\sum_{i=0}^{\lexp{d+1}}\binom{d+1-i}{k+1-i}
  \left(\binom{n_1-d-2+i}{i}+\binom{n_2-d-2+i}{i}\right),
\end{equation}
where $n_i=f_0(P_i)$, $i=1,2$, and $C_d(n)$ stands for the cyclic
$d$-polytope with $n$ vertices. The bounds in
\eqref{equ:2p-any-d} have been shown to be tight for any $d\ge{}3$ and
for all $1\le{}k\le{}d$, and, clearly, match the corresponding
bounds for 2- and 3-polytopes (cf. rel. \eqref{equ:ub3}),
as well as the expressions in \eqref{equ:trivial-ub} for $r=2$ and for
all $0\le{}k\le{}\lexp{d+1}-2$. Notice that the tightness of relations
\eqref{equ:2p-any-d}, implies that the trivial upper bounds in
\eqref{equ:trivial-ub} are also tight for $d\ge{}3$, $r=2$ and
$k\le\lexp{d+1}-2$ (as opposed to $d\ge{}4$, $r=2$ and
$0\le{}k\le{}\lexp{d}-2$).

In this paper, we show that the trivial upper bound
\eqref{equ:trivial-ub} is attained for a wider range of $d$, $r$ and
$k$ than those proved by Fukuda and Weibel \cite{fw-fmacp-07}. More
precisely, we prove that for any $d\ge{}3$, $2\le{}r\le{}d-1$ and for
all $0\le{}k\le{}\lexp{d+r-1}-r$, there exist $r$ neighborly
$d$-polytopes $\range{P}$, for which the number of $k$-faces of their
Minkowski sum attains the trivial upper bound. Our approach is based
on what is known as the \emph{Cayley trick} for Minkowski sums. Let
$\VV_i$ be the vertex set of $P_i$, $1\le{}i\le{}r$, and let
$\range{\mb{b}}$ be an affine basis of $\reals^{r-1}$.
The \emph{Cayley embedding} $\cC(\range{\VV})$, in
$\reals^{r-1}\times\reals^d$, of the vertex sets $\range{\VV}$,
is defined as
$\cC(\range{\VV})=\cup_{i=1}^{r}\{(\mb{b}_i,\mbv)\mid{}\mbv\in\VV_i\}$.
The Minkowski sum, $\MS$, can then be viewed as the intersection of
the \emph{Cayley polytope} $\pP=\conv(\cC(\range{\VV}))$ with an
appropriately defined $d$-flat $\Wavg$ of $\reals^{r-1}\times\reals^d$. 
We exploit this observation in two steps. 
We first construct a set of $r$ $(d-r+1)$-polytopes
$\range{Q}$, with $\range{n}$ vertices, respectively, embedded in
appropriate subspaces of $\reals^d$. The polytopes $\range{Q}$ are
constructed in such a way so that the number of $(k-1)$-faces of the set
$\wW_\qQ$ is equal to $\ub{k}$ for all $r\le{}k\le{}\lexp{d+r-1}$,
where $\wW_\qQ$ is the set of faces of the Cayley polytope $\qQ$ that
have non-empty 
intersection with $\Wavg$. We then perturb, via a single perturbation
parameter $\zeta$, the vertices of the $Q_i$'s to get a set of $r$
full-dimensional (\ie $d$-dimensional) neighborly polytopes $\range{P}$;
we next consider the Cayley polytope $\pP$ of the $P_i$'s, and show that is
possible to choose a small positive value for $\zeta$, so that the
number of $(k-1)$-faces of $\wW_\pP$ is equal to $\ub{k}$ for all
$r\le{}k\le{}\lexp{d+r-1}$, where $\wW_\pP$ is the set of faces of
$\pP$ with non-empty intersection with $\Wavg$. Our tight lower bound
then follows from the fact that the $(k-1)$-faces of $\wW_\pP$ are in
one-to-one correspondence with the $(k-r)$-faces of $\MS$, for all
$r\le{}k\le{}d+r-1$.

Beyond extending, with respect to $d$, $r$ and $k$, the range of tightness
of the trivial upper bound in \eqref{equ:trivial-ub}, our lower bound
construction possesses some additional interesting characteristics:
\begin{enumerate}
%\begin{smallenum}
  % 
\item[1.]
  It gives, as a special case, Fukuda and Weibel's
  tight bound on the number of vertices of the Minkowski sum of $r$
  $d$-polytopes for $d\ge{}3$ and $2\le{}r\le{}d-1$
  (cf. \cite[Theorem 2]{fw-fmacp-07}).
\item[2.]
  It constitutes a generalization of
  the lower bound construction used in \cite{kt-mnfms-11b,kt-mnfms-12}
  to prove the tightness of relation \eqref{equ:2p-any-d} for 
  $k=\lexp{d+1}-2$, and any odd $d\ge{}3$.
\item[3.]
  For any fixed dimension $d\ge{}3$ and any $2\le{}r\le{}d-1$,
  it implies a tight bound on the complexity of the Minkowski sum of $r$
  $d$-polytopes, when the polytopes have (asymptotically) the same
  number of vertices. Notice that the complexity of the Minkowski sum of
  $r$ $d$-polytopes is bounded from above by the complexity of their
  Cayley polytope, which is in $O(n^{\lexp{d+r-1}})$ when we
  assume that the polytopes have $\Theta(n)$ vertices.
  On the other hand, if $n_i=\Theta(n)$, $1\le{}i\le{}r$, our
  construction yields $f_k(\MS)=\Theta(n^{k+r})$,
  $0\le{}k\le{}\lexp{d+r-1}-r$, and, in particular,
  $f_{\lexp{d+r-1}-r}(\MS)=\Theta(n^{\lexp{d+r-1}})$.
  % Hence, for any fixed $d\ge{}3$ and any $2\le{}r\le{}d-1$, the
  % worst-case complexity of the Minkowski sum of $r$ $d$-polytopes is in
  % $\Theta(n^{\lexp{d+r-1}})$, if all polytopes have $\Theta(n)$ vertices.
  % 
\item[4.]
  It gives the maximum possible ranges of $d$, $r$ and $k$ for which
  the $k$-faces of the Minkowski sum of $r$ $d$-polytopes is equal to
  $\ub{k+r}$. If, on the contrary, the trivial upper bound was attained
  for some $k>\lexp{d+r-1}-r$, we would have that the complexity of
  the Minkowski sum of $r$ $n$-vertex $d$-polytopes would be in
  $\Omega(n^{\lexp{d+r-1}+1})$. %, if all polytopes had $n$ vertices.
  For $2\le{}r\le{}d-1$, this directly contradicts with the discussion
  in the previous item, while for $r\ge{}d\ge{}3$, it is known 
  that the complexity of the Minkowski sum of $r$ $n$-vertex
  $d$-polytopes is in $O(r^{d-1}n^{d-1})$ (cf. \cite{w-mfmsl-11}).
%\end{smallenum}
\end{enumerate}

Finally, we believe that our result is \emph{optimal} in the sense
that, for any $d\ge{}3$ and any $2\le{}r\le{}d-1$, the Minkowski sum
of the $r$ $d$-polytopes $\range{P}$ that we construct, has
the maximum possible number of $k$-faces for all $0\le{}k\le{}d-1$.
This has been proved to be true for the case of two $d$-polytopes and
for any odd $d\ge{}3$ (cf. \cite{kt-mnfms-11b,kt-mnfms-12}),
while it is straightforward to show that it also holds true for the
case of two $d$-polytopes and any even $d\ge{}4$.

The remaining sections of this paper are as follows.
In Section \ref{sec:prelim} we give some definitions, describe the
Cayley trick, and discuss its consequences that are relevant to our
results.
In Section \ref{sec:lb} we present the construction that establishes
the tightness of the trivial upper bound for $d\ge{}3$,
$2\le{}r\le{}d-1$ and $0\le{}k\le\lexp{d+r-1}-r$. We conclude with
Section \ref{sec:concl}, where we discuss our results and state
directions for future research. 
\section{Preliminaries}
\label{sec:prelim}

A \emph{convex polytope}, or simply \emph{polytope}, $P$ in
$\reals^d$ is the convex hull of a finite set of points $V$ in
$\reals^d$, called the \emph{vertex set} of $P$.
%\hide{A polytope $P$ can equivalently be described as the intersection of
%all the closed halfspaces containing $V$.}
A \emph{face} of $P$ is the intersection of $P$ with a hyperplane
for which the polytope is contained in one of the two closed
halfspaces delimited by the hyperplane.
The dimension of a face of $P$ is the dimension of its affine hull.
A $k$-face of $P$ is a $k$-dimensional face of $P$.
We consider the polytope itself as a trivial $d$-dimensional face; all
the other faces are called \emph{proper} faces. We use the term
\emph{$d$-polytope} to refer to a polytope the trivial face of which
is $d$-dimensional.
For a $d$-polytope $P$, the $0$-faces of $P$ are its
\emph{vertices}, \hide{the $1$-faces of $P$ are its \emph{edges}, the
$(d-2)$-faces of $P$ are called \emph{ridges},} while the
$(d-1)$-faces are called \emph{facets}.
For $0\le{}k{}\le{}d$ we denote by $f_k(P)$ the number of $k$-faces
of $P$.
Note that every $k$-face $F$ of $P$ is also a $k$-polytope whose
faces are all the faces of $P$ contained in $F$.
Finally, a $d$-polytope $P$ is called \emph{$k$-neighborly}, if every
subset of its vertices of size at most $k$ defines a face of
$P$. The maximum possible level of \emph{neighborliness} of a
$d$-polytope is $\lexp{d}$.
A $\lexp{d}$-neighborly $d$-polytope is simply referred to as
\emph{neighborly}.
%%%%%%

\hide{
A \emph{polytopal complex} $\cC$ is a finite collection of polytopes
in $\reals^d$ such that
(i) $\emptyset\in\cC$,
(ii) if $P\in\cC$ then all the faces of $P$ are also in $\cC$ and
(iii) the intersection $P\cap{}Q$ for two polytopes $P$ and $Q$ in
$\cC$ is a face of both $P$ and $Q$. 
The dimension $\dim(\cC)$ of $\cC$ is the largest dimension of a
polytope in $\cC$.
A polytopal complex is called \emph{pure} if all its maximal (with
respect to inclusion) faces have the same dimension. 
In this case the maximal faces are called the \emph{facets} of
$\cC$. We use the term \emph{$d$-complex} to refer to a
polytopal complex whose maximal faces are $d$-dimensional (\ie the
dimension of $\cC$ is $d$).
A polytopal complex is simplicial if all its faces are simplices. 
Finally, a polytopal complex $\cC'$ is called a \emph{subcomplex}
of a polytopal complex $\cC$ if all faces of $\cC'$ are also faces
of $\cC$.

One important class of polytopal complexes arise from polytopes.
More precisely, a $d$-polytope $P$, together with all its
faces and the empty set, form a $d$-complex, denoted
by $\cC(P)$. The only maximal face of $\cC(P)$, which is clearly
the only facet of $\cC(P)$, is the polytope $P$ itself.
Moreover, all proper faces of $P$ form a pure $(d-1)$-complex,
called the \emph{boundary complex} $\cC(\bx{}P)$, or simply
$\bx{}P$ of $P$. The facets of $\bx{}P$ are just the facets of $P$,
and its dimension is, clearly, $\dim(\bx{}P)=\dim(P)-1=d-1$.
}%\hide

\hide{
The $f$-vector $\mb{f}(P)=(f_{-1}(P),f_0(P),\ldots,f_{d-1}(P))$
of a $d$-polytope $P$ (or its boundary complex $\bx{}P$) is
defined as the $(d+1)$-dimensional vector consisting of the number
$f_k(P)$ of $k$-faces of $P$, $-1\le{}k\le{}d-1$, where $f_{-1}(P)=1$
refers to the empty set.
}%\hide

%\input{mn}

%\subsection{The Cayley trick}
\paragraph{The Cayley trick.}
\label{sec:cayleytrick}

Let $\range{P}$ be $r$ $d$-polytopes with vertex sets $\range{\VV}$,
respectively. Let $\range{\mb{b}}$ be an affine basis of
$\reals^{r-1}$, and call $\mu_i:\reals^d\to\reals^{r-1}\times\reals^d$,
the affine inclusion given by $\mu_i(\mb{x})=(\mb{b}_i,\mb{x})$. 
The \emph{Cayley embedding} $\cC(\range{\VV})$ of the point sets
$\range{\VV}$ is defined as $\cC(\range{\VV})=\cup_{i=1}^r\mu_i(\VV_i)$.
The polytope corresponding to the convex hull
$\conv(\cC(\range{\VV}))$ of the Cayley embedding
$\cC(\range{\VV})$ of $\range{\VV}$ is typically
referred to as the \emph{Cayley polytope} of $\range{P}$.

Let $\mblam=(\range{\lambda})$ be a \emph{weight vector},
\ie $\lambda_i>0$, $1\le{}i\le{}r$, and $\sum_{i=1}^r\lambda_i=1$. The
following lemma, known as \emph{the Cayley trick for Minkowski sums},
relates the $\mblam$-weighted Minkowski sum of the polytopes
$\range{P}$ with the Cayley polytope of these polytopes.
%(we have appropriately rephrased the statement of
%the corresponding lemma in \cite{hrs-ctlsb-00} to match with our
%definitions and notation):

\begin{lemma}[{\cite[Lemma 3.2]{hrs-ctlsb-00}}]
  \label{lem:cayley-embedding}
  Let $\range{P}$ be $r$ $d$-polytopes with vertex sets
  $\range{\VV}\subset\reals^d$. Moreover, let 
  $\mblam=(\range{\lambda})$ be a weight vector, and
  $W(\mblam):=\{\lambda_1\mb{b}_1+\cdots+\lambda_r\mb{b}_r\}\times\reals^d\subset\reals^{r-1}\times\reals^d$.
  Then, the $\mblam$-weighted Minkowski sum
  $\WMS\subset\reals^d$ has the
  following representation as a section of the Cayley embedding
  $\cC(\range{\VV})$ in $\reals^{r-1}\times\reals^d$:
  {\allowdisplaybreaks
    \begin{align*}
      \WMS&\cong\cC(\range{\VV})\wedge{}W(\mblam)\\
      :=\Big\{&\conv
      \{(\mb{b}_1,\mbv_1),(\mb{b}_2,\mbv_2),\ldots,(\mb{b}_r,\mbv_r)\}
      \cap{}W(\mblam)\,:\\
      &(\mb{b}_1,\mbv_1),(\mb{b}_2,\mbv_2),\ldots,(\mb{b}_r,\mbv_r)
      \in{}\cC(\range{\VV})\Big\}.
    \end{align*}}%
  Moreover, $F$ is a facet of $\WMS$ if and only if it is of the form
  $F=F'\wedge{}W(\mblam)$ for a facet $F'$ of $\cC(\range{\VV})$
  containing at least one point $(\mb{b}_i,\mbv_i)$ for all $1\le{}i\le{}r$.
\end{lemma}

As described in \cite[Corollary 3.7]{hrs-ctlsb-00}, the
$\mblam$-weighted Minkowski sums of $r$ polytopes
$\range{P}$ have isomorphic posets of subdivisions for
different values of $\mblam$. This implies that,
%that the
%$\mblam$-weighted Minkowski sums of $\range{P}$, for any
%two different values $\mblam'$ and $\mblam''$ of 
%$\mblam$, are combinatorially equivalent. In particular,
for any weight vector $\mblam$, the $\mblam$-weighted Minkowski
sum is equivalent to the $\lambdaavg$-weighted Minkowski sum, where
$\lambdaavg$ is the averaging weight vector:
$\lambdaavg=(\frac{1}{r},\frac{1}{r},\ldots,\frac{1}{r})$. On
the other hand, the $\lambdaavg$-weighted Minkowski sum
is nothing but a scaled version of the unweighted Minkowski sum
$\MS$, \ie for any weight vector
$\mblam$, the $\mblam$-weighted Minkowski sum is
combinatorially equivalent to the unweighted Minkowski sum. In that
respect, in the rest of the paper we only consider the (unweighted)
Minkowski sum of $\range{P}$, while our results carry over to
$\mblam$-weighted Minkowski sums, for any weight vector
$\mblam$.

Let $\pP$ be the Cayley polytope of $\range{P}$, and call $\wW_\pP$
the set of faces of $\pP$ that have non-empty intersection with the
$d$-flat $\Wavg=W(\lambdaavg)$. A direct consequence of
Lemma \ref{lem:cayley-embedding} is a bijection between the $(k-1)$-faces
of $\wW_\pP$ and the $(k-r)$-faces of $\MS$, for
$r\le{}k\le{}d+r-1$. This further implies that
\begin{equation}\label{equ:fkW}
  f_{k-1}(\wW_\pP)=f_{k-r}(\MS),\qquad r\le{}k\le{}d+r-1.
\end{equation}

\section{Lower bound construction}
\label{sec:lb}

Given a set $\SSS$ and a partition $\ptn[S]$ of $\SSS$ into $r$ subsets
$\range{S}$, we say that $S$ is a {\em\spansub[{\ptn[S]}]} of $\SSS$
if $S\cap{}S_i\ne\emptyset$ for all $1\le{}i\le{}r$.
Assuming that $n_i$ is the cardinality of $S_i$, the number of
\spansubs[{\ptn[S]}] of $\SSS$ of size $k\ge{}r$ is equal to $\ub{k}$
(cf. \eqref{equ:phi-def}).

In what follows we assume that $d\ge{}3$ and $2\le{}r\le{}d-1$.
We denote by $\range[r-1]{\me}$ the standard basis of $\reals^{r-1}$,
while we use $\me_0$ to denote the zero vector in
$\reals^{r-1}$. Notice that the vectors $\me_0,\range[r-1]{\me}$
form an affine basis of $\reals^{r-1}$.
Consider a set of $r$ $d$-polytopes $\range{P}$, where
$2\le{}r\le{}d-1$. Let $\VV_i$ be the vertex set of $P_i$,
$1\le{}i\le{}r$, let $\VV=\cup_{i=1}^{r}\VV_i$, and call $\ptn$
the partition of $\VV$ into its subsets $\range{\VV}$.
Let $\pP$ be the Cayley polytope of $\range{P}$,
where, in order to perform the Cayley embedding, we have chosen
$\me_0,\me_1,\ldots,\me_{r-1}$ as the affine basis of
$\reals^{r-1}$. Let $\Wavg$ denote the $d$-flat
\begin{equation}\label{equ:Wflat-def}
  \Wavg=
  \{\tfrac{1}{r}\me_0+\tfrac{1}{r}\me_1+\cdots\tfrac{1}{r}\me_{r-1}\}
  \times\reals^d
\end{equation}
of $\reals^{r-1}\times\reals^d$. Call $\wW_\pP$ the set of faces of
$\pP$ that have non-empty intersection with $\Wavg$.
As described in the previous section, the intersection of $\pP$ with
$\Wavg$ is combinatorially equivalent to the Minkowski sum
$\MS$, and, in particular, the $k$-faces of $\MS$ are in one-to-one
correspondence with the $(k+r-1)$-faces of $\wW_\pP$,
$0\le{}k\le{}d-1$. In fact, the faces of $\wW_\pP$ are precisely the
faces of $\pP$ whose vertex sets are \spansubs of $\VV$.
In view of relation \eqref{equ:fkW}, maximizing the value of
$f_{k-r}(\MS)$ is equivalent to maximizing the value of
$f_{k-1}(\wW_\pP)$, where $2\le{}r\le{}d-1$ and $r\le{}k\le{}\dup$.
In this section we exploit this observation, so as to construct a set of
$r$ $d$-polytopes $\range{P}$, with $\range{n}$ vertices,
respectively, for which the number of $(k-1)$-faces of
$\wW_\pP$ attains its maximal possible value $\ub{k}$, for all
$r\le{}k\le\lexp{\dup}$.
%Using relation \eqref{equ:fkW}, we then conclude that, for the
%polytopes $\range{P}$ mentioned above, the number of
%$k$-faces of $\MS$ attains its maximum possible
%value $\ub{k+r}$, for all $0\le{}k\le\lexp{\dup}-r$.

Before getting into the technical details, we first outline our
approach. We start by considering the $\pddo$-dimensional moment curve,
which we embed in $r$ distinct subspaces of $\reals^d$. We
consider the $r$ copies of the $\pddo$-dimensional moment curve as
different curves, and we perturb them appropriately, so that they
become $d$-dimensional moment-like curves. The perturbation is
controlled via a non-negative parameter $\zeta$, which will be chosen
appropriately. We then choose points on these $r$ moment-like curves,
all parameterized by a positive parameter $\tau$, which will again be
chosen appropriately.
We call $\range{P}$ the $r$ $d$-polytopes we get by considering the
points on each moment-like curve, $\pP$ the Cayley polytope of
$\range{P}$, and $\wW_\pP$ the set of faces of $\pP$ that have
non-empty intersection with $\Wavg$.
For these polytopes we show that the number of $(k-1)$-faces of
$\wW_\pP$, where $r\le{}k\le{}\lexp{\dup}$, becomes equal to $\ub{k}$
for small enough positive values of $\zeta$ and $\tau$.

At a more technical level, the proof that $f_{k-1}(\wW_\pP)=\ub{k}$, for all
$r\le{}k\le\lexp{\dup}$, is performed in two steps. We first consider
the cyclic $\pddo$-polytopes $\range{Q}$, embedded in appropriate
subspaces of $\reals^d$. The $Q_i$'s are the \emph{unperturbed},
with respect to $\zeta$, versions of the $d$-polytopes $\range{P}$
(\ie the polytope $Q_i$ is the polytope we get from $P_i$,
when we set $\zeta$ equal to zero). We consider the Cayley
polytope $\qQ$ of $\range{Q}$, seen as polytopes in
$\reals^d$, and we focus on the set $\wW_\qQ$ of faces of $\qQ$, that
are the faces of $\qQ$ intersected by $\Wavg$. Noticing that the
polytopes $\range{Q}$ are parameterized by the parameter $\tau$,
we show that there exists a sufficiently small positive value for
$\tau$, for which the number of $(k-1)$-faces of $\wW_\qQ$ is equal to
$\ub{k}$. Having chosen the appropriate value for $\tau$,
which we denote by $\tau^\star$, we then consider the
polytopes $\range{P}$ (with $\tau$ set to $\tau^\star$), and show that
for sufficiently small $\zeta$, $f_{k-1}(\wW_{\pP})$ is equal to
$\ub{k}$.

%%%%%%%%%%%%%%%%%%%%%%%%%%%%%%%%%%%%%%%%%%%%%%%%%%%%%%%%%%%%%%%%%%%%%%%%
% technical lemma -- START
We start off with a technical lemma and sketch its proof. The detailed
proof may be found in Section \ref{app:signdet} of the Appendix.

\techlemma{\label{lem:sign-generic-det}}
\savecounter{\thetheorem}
% technical lemma -- END
%%%%%%%%%%%%%%%%%%%%%%%%%%%%%%%%%%%%%%%%%%%%%%%%%%%%%%%%%%%%%%%%%%%%%%%%

\begin{proof}[Sketch of proof]
  Let $K_i=\sum_{j=1}^{i}\kappa_j$, and let
  $\mb{c}=(\range[n]{\mb{c}})$, where $\mb{c}_i$ is the column vector
  corresponding to the columns $K_{i-1}+1$ to $K_i$ of $\mddet$,
  $1\le{}i\le{}n$ (by convention $K_0=0$).
  Let $\mb{r}=(\range[n]{\mb{r}})$ denote a row vector of
  $\mddet$, where $\mb{r}_i\in\naturals^{\kappa_i}$.
  We use Laplace's Expansion Theorem, to expand $\mddet$ along its first
  $\kappa_1$ columns, then along its next $\kappa_2$ columns, etc. This
  produces an expansion for $\mddet$, consisting of
  $\prod_{i=1}^{n}\binom{K-K_{i-1}}{\kappa_i}$ terms, where each term
  corresponds to a different choice for $\mb{r}$.
  More precisely, each term is, up to a sign, the product of $n$ minors
  $M_i(\tau)$ of $\mddet$ of size $\kappa_i\times\kappa_i$, $1\le{}i\le{}n$,
  where the columns (\resp rows) of $M_i(\tau)$ are the columns
  (\resp rows) in $\mb{c}_i$ (\resp $\mb{r}_i$).
  Among these terms, there are
  $\prod_{i=1}^{n}\binom{K-2n-K_{i-1}}{\kappa_i-2}$ non-vanishing
  terms, each of which is of the form
  $(-1)^{\sigma(\mb{r},\mb{c})+N}\tau^{\theta(\mb{r},\mb{c})}
  \prod_{i=1}^{n}D_i$, where $D_i$ is a generalized Vandermonde determinant.
  We proceed by identifying the unique, with respect to $\mb{r}$,
  term in the expansion, for which $\theta(\mb{r},\mb{c})$ is
  minimal. Denoting by $\mb{\rho}$ the row vector for which this
  minimal value is attained, we show that $\sigma(\mb{\rho},\mb{c})+N$
  is even. We, thus, get $\mddet=\tau^{\theta(\mb{\rho},\mb{c})}
  \prod_{i=1}^{n}D_i+O(\tau^{\theta(\mb{\rho},\mb{c})+1})$.
  Our result then follows by taking the limit
  $\lim_{\tau\to{}0+}\frac{\mddet}{\tau^{\theta(\mb{\rho},\mb{c})}}$,
  and by noticing that the determinants $D_i$, $1\le{}i\le{}n$,
  are strictly positive.
\end{proof}

Let $\mc(t)$, $t>0$, be the $\pddo$-dimensional
moment curve, \ie $\mc(t)=(t,t^2,\ldots,t^{\ddo})$. We are going to
call $\mc_i(t)$, $1\le{}i\le{}r$, the curve $\mc(t)$ embedded in the
$\pddo$-flat $F_i$ of $\reals^d$, where
{
  \setlength\abovedisplayskip{7pt plus 3pt minus 7pt}
  \setlength\belowdisplayskip{7pt plus 3pt minus 7pt}
  \begin{equation}\label{equ:Fi}
    F_i=\{x_j=0\mid{}1\le{}j\le{}r\text{ and }j\ne{}i\},
    \qquad 1\le{}i\le{}r,
  \end{equation}}%
such that the first coordinate of $\mc(t)$ becomes the $i$-coordinate
of $\mc_i(t)$, whereas, for $2\le{}j\le{}\ddo$, the $j$-th coordinate
of $\mc(t)$ becomes the $(j+r-1)$-coordinate of $\mc_i(t)$.
In other words:
{\allowdisplaybreaks
  \setlength\abovedisplayskip{7pt plus 3pt minus 7pt}
  \setlength\belowdisplayskip{0pt plus 3pt minus 7pt}
  \begin{align*}
    \mc_1(t)&=(\overbrace{t,0,0,0,\ldots,0,0}^{r},
    \overbrace{t^2,\ldots,t^{\ddo}}^{d-r}),\\
    \mc_2(t)&=(0,t,0,0,\ldots,0,0,t^2,\ldots,t^{\ddo}),\\
    \mc_3(t)&=(0,0,t,0,\ldots,0,0,t^2,\ldots,t^{\ddo}),\\
    &\ \,\vdots\\
    \mc_r(t)&=(0,0,0,0,\ldots,0,t,t^2,\ldots,t^{\ddo}).
  \end{align*}}%
We next \emph{perturb} the vanishing coordinates of $\mc_i(t)$,
$1\le{}i\le{}r$, to get the $d$-dimensional curve $\mc_i(t;\zeta)$ as
follows: the first (from the left) vanishing coordinate of $\mc_i(t)$
becomes $\zeta t^{d-r+2}$, the second vanishing coordinate of $\mc_i(t)$
becomes $\zeta t^{d-r+3}$, etc., and, finally, the last vanishing
coordinate of $\mc_i(t)$ becomes $\zeta t^d$:
\begin{align*}
  \mc_1(t;\zeta)&=(\overbrace{t,\zeta t^{d-r+2},\zeta t^{d-r+3},
    \zeta t^{d-r+4},\ldots,\zeta t^{d-1},\zeta t^d}^{r},
  \overbrace{t^2,\ldots,t^{\ddo}}^{d-r}),\\
  \mc_2(t;\zeta)&=(\zeta t^{d-r+2},t,\zeta t^{d-r+3},
  \zeta t^{d-r+4},\ldots,\zeta t^{d-1},\zeta t^d,t^2,\ldots,t^{\ddo}),\\
  \mc_3(t;\zeta)&=(\zeta t^{d-r+2},\zeta t^{d-r+3},t,
  \zeta t^{d-r+4},\ldots,\zeta t^{d-1},\zeta t^d,t^2,\ldots,t^{\ddo}),\\
  &\ \,\vdots\\
  \mc_r(t;\zeta)&=(\zeta t^{d-r+2},\zeta t^{d-r+3},\zeta t^{d-r+4},
  \ldots,\zeta t^{d-1},\zeta t^d,t,t^2,\ldots,t^{\ddo}),
\end{align*}
where $\zeta\ge{}0$ is the perturbation parameter (clearly, for
$\zeta=0$, $\mc_i(t;\zeta)$ reduces to $\mc_i(t)$).
Denote by $\mcc_i(t)$ (\resp $\mcc_i(t;\zeta)$) the Cayley embedding
of $\mc_i(t)$ (\resp $\mc_i(t;\zeta)$), \ie:
\begin{equation*}
  \begin{aligned}
    \mcc_i(t)&=\mu_i(\mc_i(t))=(\me_{i-1},\mc_i(t)),\\
    \mcc_i(t;\zeta)&=\mu_i(\mc_i(t;\zeta))=(\me_{i-1},\mc_i(t;\zeta)),
  \end{aligned}
  \qquad t>0, \qquad 1\le{}i\le{}r.
\end{equation*}

For each $i$ with $1\le{}i\le{}r$, choose $n_i$ real numbers
$\alpha_{i,j}$, $j=1,\ldots,n_i$, such that
$0<\alpha_{i,1}<\alpha_{i,2}<\ldots<\alpha_{i,n_i}$.
Let $\tau$ be a strictly positive parameter that is determined
below, and choose $r$ non-negative integers $\range{\nu}$, such that
$\nu_1>\nu_2>\cdots>\nu_r=0$.
Let $\UU_i$, $1\le{}i\le{}r$, be the $\pddo$-dimensional point
sets:
\begin{equation*}%\label{equ:Ui}
  \UU_i=\{\mc_i(\alpha_{i,1}\tau^{\nu_i}),\mc_i(\alpha_{i,2}\tau^{\nu_i}),\ldots,
  \mc_i(\alpha_{i,n_i}\tau^{\nu_i})\}.
\end{equation*}
Since $\UU_i$ consists of points on the $\pddo$-dimensional moment curve
$\mc(t)$, embedded in the $\pddo$-flat $F_i$ of
$\reals^d$, the $\pddo$-polytope $Q_i=\conv(\UU_i)$ is the
cyclic $\pddo$-polytope embedded in $F_i$ (cf. \eqref{equ:Fi}).

Let $\UUC_i=\mu_i(\UU_i)$, $1\le{}i\le{}r$, $\UU=\cup_{i=1}^{r}\UU_i$, 
$\UUC=\cup_{i=1}^{r}\UUC_i$, and denote by $\ptn[U]$ (\resp
$\ptnc[U]$) the partition of $\UU$ (\resp $\UUC$) into its subsets
$\range{\UU}$ (\resp $\range{\UUC}$).
Let $\qQ=\conv(\UUC)$ be the Cayley polytope
of $\range{Q}$, and let $\wW_{\qQ}$ be the set of faces of
$\qQ$ with non-empty intersection with $\Wavg$, i.e.,
$\wW_{\qQ}$ consists of all the faces of $\qQ$, the vertex set
of which is a \spansub[{\ptnc[U]}] of $\UUC$.
The following lemma establishes the first step towards our construction.

\begin{lemma}\label{lem:Qstar}
  There exists a sufficiently small positive value $\tau^\star$ for
  $\tau$, such that
  \[  f_{k-1}(\wW_\qQ)=\ub{k},\qquad r\le{}k\le{}\ltexp{\dup}.\]
\end{lemma}

\begin{proof}
  To simplify the notation used in the proof, we
  identify $\UUC_i$, $\UUC$ and $\ptnc[U]$ with their pre-images under
  the Cayley embedding, namely, $\UU_i$, $\UU$ and $\ptn[U]$, respectively.

  Let $\alpha_{i,j}^\epsilon=\alpha_{i,j}+\epsilon$,
  $t_{i,j}=\alpha_{i,j}\tau^{\nu_i}$,
  $t_{i,j}^\epsilon=\alpha_{i,j}^\epsilon\tau^{\nu_i}$, where
  $1\le{}j\le{}n_i$, $1\le{}i\le{}r$, and $\epsilon>0$. The value of
  $\epsilon$ is chosen such that $\alpha_{i,j}^\epsilon<\alpha_{i,j+1}$,
  for all $1\le{}j<n_i$, and for all $1\le{}i\le{}r$.
  Finally, let $M$ be a positive real number such that
  $M>\alpha_{r,n_r}^\epsilon=\alpha_{r,n_r}^\epsilon\tau^{\nu_r}
  =t_{r,n_r}^\epsilon$ (recall that $\nu_r=0$).

  Choose a \spansub[{\ptn[U]}] $U$ of $\UU$ of size $k$, and
  denote by $k_i$ the cardinality of $U_i=U\cap{}\UU_i$;
  clearly, $\sum_{i=1}^{r}k_i=k$.
  Let $\mcc_i(t_{i,j_{i,1}}),\mcc_i(t_{i,j_{i,2}}),\ldots,\mcc_i(t_{i,j_{i,k_i}})$,
  be the vertices in $U_i$, where $j_{i,1}<j_{i,2}<\ldots<j_{i,k_i}$ and
  $1\le{}i\le{}r$.
  Let $\mb{x}=(\range[\dup]{x})$ and define the
  $(d+r)\times(d+r)$ determinant $H_U(\mb{x})$ as
  follows\footnote{\Hexample}:
  \begin{equation}\label{equ:H-def}
    \begin{aligned}
      H_U(\mb{x})=
      (-1)^R\,\bigg|&
        \begin{smallmatrix}
          1&1&1&\cdots&1&1&1&1&\cdots&1&1\\[3pt]
          \mb{x}&\mcc_1(t_{1,j_{1,1}})&\mcc_1(t_{1,j_{1,1}}^\epsilon)&\cdots
          &\mcc_1(t_{1,j_{1,k_1}})&\mcc_1(t_{1,j_{1,k_1}}^\epsilon)
          &\mcc_2(t_{2,j_{2,1}})&\mcc_2(t_{2,j_{2,1}}^\epsilon)&\cdots
          &\mcc_2(t_{2,j_{2,k_2}})&\mcc_2(t_{2,j_{2,k_2}}^\epsilon)
        \end{smallmatrix}\\[6pt]
        &
        \begin{smallmatrix}
          \cdots&1&1&\cdots&1&1\\[3pt]
          \cdots&\mcc_{r-1}(t_{r-1,j_{r-1,1}})&\mcc_{r-1}(t_{r-1,j_{r-1,1}}^\epsilon)
          &\cdots
          &\mcc_{r-1}(t_{r-1,j_{r-1,k_{r-1}}})
          &\mcc_{r-1}(t_{r-1,j_{r-1,k_{r-1}}}^\epsilon)\\
        \end{smallmatrix}\\[6pt]
        &
        \begin{smallmatrix}
          1&1&\cdots&1&1&1&1&\cdots&1\\[3pt]
          \mcc_r(t_{r,j_{r,1}})&\mcc_r(t_{r,j_{r,1}}^\epsilon)
          &\cdots&\mcc_r(t_{r,j_{r,k_r}})&\mcc_r(t_{r,j_{r,k_r}}^\epsilon)
          &\mcc_r(M)&\mcc_r(2M)
          &\cdots&\mcc_r((\dup-2k)M)
        \end{smallmatrix}
      \bigg|,
      \end{aligned}
  \end{equation}
  where $R=\frac{r(r-1)}{2}$.
  We can alternatively describe $H_U(\mb{x})$ as follows:
  \begin{enumerate}
  \item
    The first column of $H_U(\mb{x})$ is $\binom{1}{\mb{x}}$.
  \item
    For $i$ ranging from $1$ to $r$, and for $\lambda$ ranging from $1$ to
    $k_i$, the next $k_i$ \emph{pairs of columns} of
    $H_U(\mb{x})$ are $\binom{1}{\mcc_i(t_{i,j_{i,\lambda}})}$ and
    $\binom{1}{\mcc_i(t_{i,j_{i,\lambda}}^\epsilon)}$.
  \item
    For $\lambda$ ranging from $1$ to $\dup-2k$,
    the last $\dup-2k$ columns of $H_U(\mb{x})$ are
    $\binom{1}{\mcc_r(\lambda M)}$. Notice that if $k=\lexp{\dup}$ and
    $\dup$ is even, this last category of columns of $H_U(\mb{x})$ does
    not exist.
  \end{enumerate}

  The equation $H_U(\mb{x})=0$ is the equation of a
  hyperplane in $\reals^{\dup}$ that passes through the points in
  $U$. We are going to show that, for any choice of $U$,
  and for all vertices $\mbu$ in $\UU\sm{}U$,
  we have $H_U(\mbu)>0$ for sufficiently small $\tau$.
  
  Suppose we have some vertex $\mbu$ of $\UU$ such that
  $\mbu\in{}\UU_i\sm{}U$. Then, $\mbu=\mcc_i(t_{i,\lambda})$,
  $t_{i,\lambda}=\alpha_{i,\lambda}\tau^{\nu_i}$, for some
  $\lambda\nin\{j_{i,1},j_{i,2},\ldots,j_{i,k_i}\}$.
  Then $H_U(\mbu)$ becomes:
  \begin{equation*}
    \begin{aligned}
      H_U(\mbu)=
      (-1)^R\,\bigg|&
        \begin{smallmatrix}
          1&1&1&\cdots&1&1&1&1&\cdots&1&1\\[3pt]
          \mcc_i(t_{i,\lambda})&\mcc_1(t_{1,j_{1,1}})&\mcc_1(t_{1,j_{1,1}}^\epsilon)
          &\cdots
          &\mcc_1(t_{1,j_{1,k_1}})&\mcc_1(t_{1,j_{1,k_1}}^\epsilon)
          &\mcc_2(t_{2,j_{2,1}})&\mcc_2(t_{2,j_{2,1}}^\epsilon)&\cdots
          &\mcc_2(t_{2,j_{2,k_2}})&\mcc_2(t_{2,j_{2,k_2}}^\epsilon)
        \end{smallmatrix}\\[6pt]
        &
        \begin{smallmatrix}
          \cdots&1&1&\cdots&1&1\\[3pt]
          \cdots&\mcc_{r-1}(t_{r-1,j_{r-1,1}})&\mcc_{r-1}(t_{r-1,j_{r-1,1}}^\epsilon)
          &\cdots
          &\mcc_{r-1}(t_{r-1,j_{r-1,k_{r-1}}})
          &\mcc_{r-1}(t_{r-1,j_{r-1,k_{r-1}}}^\epsilon)\\
        \end{smallmatrix}\\[6pt]
        &
        \begin{smallmatrix}
          1&1&\cdots&1&1&1&1&\cdots&1\\[3pt]
          \mcc_r(t_{r,j_{r,1}})&\mcc_r(t_{r,j_{r,1}}^\epsilon)&
          &\cdots&\mcc_r(t_{r,j_{r,k_r}})&\mcc_r(t_{r,j_{r,k_r}}^\epsilon)
          &\mcc_r(M)&\mcc_r(2M)
          &\cdots&\mcc_r((\dup-2k)M)
        \end{smallmatrix}
      \bigg|.
      \end{aligned}
  \end{equation*}
  Observe now that we can transform $H_U(\mbu)$ in the form
  of the determinant $\mddet$ of Lemma \ref{lem:sign-generic-det},
  where $n\sub{}r$, $\kappa_i\sub{}2k_i+1$, $\kappa_r\sub{}2k_r+\dup-2k$, 
  $\kappa_j\sub{}2k_j$ for $1\le{}j<r$ and $j\ne{}i$,
  and $\beta_j\sub{}\nu_j$ for $1\le{}j\le{}r$, by means of the
  following determinant transformations:
  \begin{enumerate}
  \item
    By subtracting rows $2$ to $r$ of $H_U(\mbu)$ from its first row.
  \item
    By shifting the first column of $H_U(\mbu)$ to the right via
    an even number of column swaps.
    More precisely, to transform $H_U(\mbu)$ in the form of
    $\mddet$, we need to shift the first column of $H_U(\mbu)$ to
    the right so that the values $t_{i,\lambda},t_{i,j_{i,1}},t_{i,j_{i,1}}^\epsilon,
    t_{i,j_{i,2}},t_{i,j_{i,2}}^\epsilon,\ldots,t_{i,j_{i,k_i}},t_{i,j_{i,k_i}}^\epsilon$
    appear consecutively in the columns of $H_U(\mbu)$ and in
    increasing order. To do that we always need an even number of
    column swaps, due to the way we have chosen $\epsilon$.
    More precisely, we first need to shift the first
    column of $H_U(\mbu)$ through the $2\sum_{j=1}^{i-1}k_j$ columns
    to its right; this is, obviously, done via an even number of
    column swaps. Consider the following cases:
    \begin{itemize}
    \item
      If $\lambda<j_{i,1}$, then we are done: $H_U(\mbu)$ is in
      the desired form.
    \item
      If $\lambda>j_{i,k_i}$, notice that due to the way we have chosen
      $\epsilon$, we have $t_{i,\lambda}>t_{i,j_{i,k_i}}^\epsilon$;
      therefore, we need to perform another $2k_i$ column swaps in order
      to shift $\binom{1}{\mcc_i(t_{i,\lambda})}$ to its final
      place, \ie to the right of $\binom{1}{\mcc_i(t_{i,j_{i,k_i}}^\epsilon)}$.
      In other words, in this case $H_U(\mbu)$ can be transformed to
      the form of $\mddet$ with a total of $2\sum_{j=1}^{i}k_j$ column
      swaps.
    \item
      Finally, if $j_{i,1}<\lambda<j_{i,k_i}$, there exists some $\xi$ with
      $1\le{}\xi<k_i$, such that $j_{i,\xi}<\lambda<j_{i,\xi+1}$.
      Since $t_{i,\lambda}>t_{i,j_{i,\xi}}^\epsilon$, due to the way we
      have chosen $\epsilon$, we need another $2\xi$
      column swaps to place $\binom{1}{\mcc_i(t_{i,\lambda})}$
      to the right of $\binom{1}{\mcc_i(t_{i,j_{i,\xi}}^\epsilon)}$.
      Hence, in this case, $H_U(\mbu)$ can be transformed to the form
      of $\mddet$ with a total of $2(\xi+\sum_{j=1}^{i-1}k_j)$ column
      swaps.
    \end{itemize}
  \end{enumerate}
  Applying now Lemma \ref{lem:sign-generic-det}, we deduce that there
  exists a value $\tau_0$ for $\tau$, such that for all
  $\tau\in(0,\tau_0)$, the determinant $H_U(\mbu)$ is strictly
  positive.

  We thus conclude that, for any specific choice of $U$, and for
  any specific vertex $\mbu\in{}\UU\sm{}U$,
  there exists some $\tau_0>0$ (cf. Lemma \ref{lem:sign-generic-det})
  that depends on $\mbu$ and $U$, such that for all
  $\tau\in(0,\tau_0)$, $H_U(\mbu)>0$.
  Since for each $k$ with $r\le{}k\le{}\lexp{\dup}$, the number
  of \spansubs[{\ptn[U]}] $U$ of size $k$ of $\UU$ is
  $\ub{k}$, while for each such subset $U$ we need
  to consider the $(\sum_{i=1}^{r}n_i-k)$ vertices in $\UU\sm{}U$, it
  suffices to consider a positive value $\tau^\star$ for $\tau$ that
  is small enough, so that all
  \begin{equation*}
    \sum_{k=r}^{\lexp{\dup}}\Big(\sum_{i=1}^{r}n_i-k\Big)\,\ub{k}
  \end{equation*}
  possible determinants $H_U(\mbu)$ are strictly positive. 
  For $\tau\sub\tau^\star$, our analysis above immediately implies that
  \emph{for each} \spansub[{\ptn[U]}] $U$ of $\UU$
  the equation $H_{U}(\mb{x})=0$, $\mb{x}\in\reals^{\dup}$, is the
  equation of a supporting hyperplane for $\qQ$ passing through the 
  vertices of $U$, and those only. In other words, every
  \spansub[{\ptn[U]}] $U$ of $\UU$, where
  $|U|=k$ and $r\le{}k\le{}\lexp{\dup}$,
  defines a $(k-1)$-face of $\qQ$, which means that 
  $f_{k-1}(\wW_{\qQ})=\ub{k}$, for all $r\le{}k\le{}\lexp{\dup}$.
\end{proof}

We assume we have chosen $\tau$ to be equal to $\tau^\star$, and
call $\UU_i^\star$, $Q_i^\star$,
$1\le{}i\le{}r$, the corresponding vertex sets and
$\pddo$-polytopes. Let $\UU^\star=\cup_{i=1}^{r}\UU_i^\star$, and call
$\qQ^\star$ the Cayley polytope of $Q_1^\star,Q_2^\star,\ldots,Q_r^\star$.
We are going to perturb the vertex sets $\UU_i^\star$ to get the vertex
sets $\VV_i$, $1\le{}i\le{}r$, by considering vertices on the curves
$\mc_i(t;\zeta)$ with $\zeta>0$, instead of the curves $\mc_i(t)$.
More precisely, define the sets $\VV_i$, $1\le{}i\le{}r$, as follows:
\begin{equation*}%\label{equ:Vi}
  \VV_i=\{\mc_i(\alpha_{i,1}(\tau^\star)^{\nu_i};\zeta),
  \mc_i(\alpha_{i,2}(\tau^\star)^{\nu_i};\zeta),\ldots,
  \mc_i(\alpha_{i,n_i}(\tau^\star)^{\nu_i};\zeta)\},
\end{equation*}
where $\zeta>0$. Let $\VV=\cup_{i=1}^{r}\VV_i$, call $\ptn$ the
partition of $\VV$ into its subsets $\range{\VV}$,
and let $P_i$ be the $d$-polytope whose vertex set is $\VV_i$.
It is easy to verify that:

\begin{claim}\label{claim:neighborly}
  For any $\zeta>0$, $P_i$ is a neighborly $d$-polytope.
\end{claim}

\begin{proof}
  Recall that the vertices of $P_i$ are taken from the moment-like curve:
  \begin{equation}
    \mc_i(t;\zeta)=(\overbrace{\zeta t^{d-r+2},\zeta t^{d-r+3},\ldots,
      \zeta t^{d-r+i}}^{i-1},t,
    \overbrace{\zeta t^{d-r+i+1},\ldots,\zeta t^{d-1},\zeta t^d}^{r-i},
    \overbrace{t^2,\ldots,t^{d-r+1}}^{d-r}).
  \end{equation}
  Let $t_{j}=\alpha_{i,j}(\tau^\star)^{\nu_i}$,
  $t_{j}^\epsilon=t_{j}+\epsilon$, $1\le{}j\le{}n_i$, and
  $M>t_{n_i}^\epsilon$, where $\epsilon$ is a small positive constant
  chosen so that $t_{j}^\epsilon<t_{j+1}$, for all $1\le{}j<n_i$.

  We will first show that $P_i$ is $d$-dimensional. Consider a subset
  $V$ of $\VV_i$ of size $d+1$, and let
  $\mc_i(t_{j_1};\zeta),\mc_i(t_{j_2};\zeta),\ldots,\mc_i(t_{j_{d+1}};\zeta)$,
  be the vertices in $V$, where $j_1<j_2<\ldots<j_{d+1}$.
  Define the $(d+1)\times(d+1)$ determinant
  $F(\zeta)$, corresponding to the volume of $\conv(V)$ in $\reals^d$:
  \begin{equation*}
    F(\zeta)=
      \begin{vmatrix}
        1&1&1&\cdots&1\\[3pt]
        \mc_i(t_{j_1};\zeta)&\mc_i(t_{j_2};\zeta)
        &\mc_i(t_{j_3};\zeta)&\cdots
        &\mc_i(t_{j_{d+1}};\zeta)
      \end{vmatrix}.
  \end{equation*}
  It is easy to verify that
  $F(\zeta)=(-1)^{i-1+(r-1)(d-r)}\zeta^{r-1}\VD{t}$, where
  $\mb{t}=(t_{j_1},t_{j_2},\ldots,t_{j_{d+1}})$; recall that $\VD{x}$
  denotes the Vandermonde determinant corresponding to the vector
  $\mb{x}$ (cf. \eqref{equ:VD}). Since the elements in $\mb{t}$ are in
  strictly increasing order, we immediately conclude that
  $\VD{t}>0$. This further implies that $F(\zeta)\ne{}0$, for any
  $\zeta>0$. Hence, the polytope $P_i$ is $d$-dimensional, since it
  contains at least one $d$-dimensional simplex.

  We will now show that $P_i$ is neighborly.
  Consider a subset $V$ of $\VV_i$ of size $k\le\lexp{d}$, and let
  $\mc_i(t_{j_1};\zeta),\mc_i(t_{j_2};\zeta),\ldots,\mc_i(t_{j_k};\zeta)$
  be the vertices of $P_i$ in $V$, where $j_1<j_2<\ldots<j_k$.
  Define the $(d+1)\times(d+1)$ determinant $H_V(\mb{x})$,
  $\mb{x}\in\reals^d$, as follows\footnote{For $d$ even and
    $k=\lexp{d}$ the columns of $H_V(\mb{x})$ corresponding to $M, 2M,
    \ldots, (d-2k)M$ do not exist.}:
  \begin{equation*}
    H_V(\mb{x})=\left|
      \begin{smallmatrix}
        1&1&1&1&1&\cdots&1&1&1&1&\cdots&1\\[3pt]
        \mb{x}
        &\mc_i(t_{j_1})&\mc_i(t_{j_1}^\epsilon)
        &\mc_i(t_{j_2})&\mc_i(t_{j_2}^\epsilon)&\cdots
        &\mc_i(t_{j_k})&\mc_i(t_{j_k}^\epsilon)
        &\mc_i(M)&\mc_i(2M)&\cdots&\mc_i((d-2k)M)
      \end{smallmatrix}\right|.
  \end{equation*}
  Observe that
  \begin{equation*}
    H_V(\mc_i(t))=(-1)^{i-1+(r-1)(d-r)}\zeta^{r-1}\VD{T},
  \end{equation*}
  where
  $\mb{T}=(t,t_{j_1},t_{j_1}^\epsilon,t_{j_2},t_{j_2}^\epsilon,\ldots,
  t_{j_k},t_{j_k}^\epsilon,M, 2M,\ldots,(d-2k)M)$.
  In other words, for any $\zeta>0$, $H_V(\mc_i(t))$ is a polynomial in
  $t$ of degree $d$, which has $d$ distinct real roots, namely,
  $t_{j_1},t_{j_2},\ldots,t_{j_k}$,
  $t_{j_1}^\epsilon,t_{j_2}^\epsilon,\ldots,t_{j_k}^\epsilon$, and
  $M, 2M,\ldots,(d-2k)M$. Observe, also, that there always exists an
  even number of roots\footnote{For $d$ even and $k=\lexp{d}$ the $d$
    real roots of $H_V(\mc_i(t))$ are $t_{j_1},t_{j_2},\ldots,t_{j_k}$,
    $t_{j_1}^\epsilon,t_{j_2}^\epsilon,\ldots,t_{j_k}^\epsilon$.}
  of $H_V(\mc_i(t))$ between $t=t_{\mu}$ and
  $t=t_{\xi}$ for any $\mu,\xi$ with $1\le{}\mu\ne\xi\le{}n_i$ and
  $\mu,\xi\nin\{\range[k]{j}\}$. This immediately implies that
  $H_V(\mc_i(t))$ has always the same sign for any $t_\ell$ with
  $\ell\nin\{\range[k]{j}\}$, which further implies that $H_V(\mbv)$
  has the same sign for any $\mbv\in{}\VV_i\sm{}V$. In other words,
  $H_V(\mb{x})=0$ is the equation of a supporting hyperplane of $P_i$,
  passing through the vertices of $V$, and those only.

  Since we have chosen $V$ arbitrarily, the same holds for any
  $V\subseteq{}\VV_i$ with $|V|=k\le{}\lexp{d}$. Thus, for any subset
  $V$ of $\VV_i$ of size $k\le\lexp{d}$, $V$ defines a $(k-1)$-face of
  $P_i$, \ie $P_i$ is neighborly.
\end{proof}

% for any $\zeta>0$, $P_i$ is a neighborly
%$d$-polytope (see Section \ref{app:neighborly} of the Appendix).

%
Let $\pP=\conv(\cC(\range{\VV}))$ be the Cayley polytope of
$\range{P}$, and let $\wW_\pP$ be the set of faces of $\pP$
with non-empty intersection with the $d$-flat $\Wavg$ of
$\reals^{r-1}\times\reals^d$ (cf. \eqref{equ:Wflat-def}), \ie
$\wW_\pP$ consists of all the faces of $\pP$, the vertex set of
which is a \spansub[{\ptnc[V]}] of $\VVC$, where
$\ptnc[V]=\{\range{\VVC}\}$, $\VVC=\cup_{i=1}^{r}\VVC_i$, and
$\VVC_i=\mu_i(\VV_i)$, for $1\le{}i\le{}r$.

The following lemma establishes the second, and final, step of our
construction.

\begin{lemma}\label{lem:Pstar}
  There exists a sufficiently small positive value $\zeta^{\lozenge}$ for
  $\zeta$, such that
  \[ f_{k-1}(\wW_\pP)=\ub{k},\qquad r\le{}k\le{}\ltexp{\dup}.\]
\end{lemma}

\begin{proof}
  As in the proof of Lemma \ref{lem:Qstar}, and in order to simplify
  the notation used in the proof,
  we identify $\VVC_i$, $\VVC$ and $\ptnc$ with their pre-images under
  the Cayley embedding, namely, $\VV_i$, $\VV$ and $\ptn$, respectively.

  Similarly to what we have done in the proof of Lemma \ref{lem:Qstar},
  let $\alpha_{i,j}^\epsilon=\alpha_{i,j}+\epsilon$,
  $t_{i,j}=\alpha_{i,j}(\tau^\star)^{\nu_i}$, 
  $t_{i,j}^\epsilon=\alpha_{i,j}^\epsilon(\tau^\star)^{\nu_i}$,
  $1\le{}j\le{}n_i$, $1\le{}i\le{}r$, and $M$ be a positive real
  number such that
  $M>\alpha_{r,n_r}^\epsilon=t_{r,n_r}^\epsilon$,
  where $\epsilon>0$ is chosen such that
  $\alpha_{i,j}^\epsilon<\alpha_{i,j+1}$, for all $1\le{}j<n_i$, and
  for all $1\le{}i\le{}r$.

  Choose a \spansub $V$ of $\VV$ of size $k$, and
  denote by $k_i$ the cardinality of $V_i=V\cap{}\VV_i$,
  $1\le{}i\le{}r$.
  Considering $\zeta$ as a small positive parameter, let
  $\mcc_i(t_{i,j_{i,1}};\zeta),\mcc_i(t_{i,j_{i,2}};\zeta),
  \ldots,\mcc_i(t_{i,j_{i,k_i}};\zeta)$
  be the vertices in $V_i$, where $j_{i,1}<j_{i,2}<\ldots<j_{i,k_i}$.
  Let $\mb{x}=(\range[\dup]{x})$ and define the
  $(d+r)\times(d+r)$ determinant $F_V(\mb{x};\zeta)$
  as\footnote{\Fexample}:
  {\allowdisplaybreaks
  \begin{equation}\label{equ:F-def}
    \begin{aligned}
      F_V(\mb{x};\zeta)=(-1)^R\,\bigg|&
        \begin{smallmatrix}
          % {@{\hsep}c@{\hsep}c@{\hsep}c@{\hsep}c@{\hsep}c@{\hsep}c@{\hsep}}
          % c@{\hsep}c@{\hsep}c@{\hsep}c@{\hsep}c@{\hsep}}
          1&1&1&\cdots&1&1\\[3pt]
          \mb{x}&\mcc_1(t_{1,j_{1,1}};\zeta)&\mcc_1(t_{1,j_{1,1}}^\epsilon;\zeta)
          &\cdots
          &\mcc_1(t_{1,j_{1,k_1}};\zeta)&\mcc_1(t_{1,j_{1,k_1}}^\epsilon;\zeta)
        \end{smallmatrix}\\[6pt]
        &
        \begin{smallmatrix}%{@{\hsep}c@{\hsep}c@{\hsep}c@{\hsep}c@{\hsep}c@{\hsep}c@{\hsep}}
          &1&1&\cdots&1&1\\[3pt]
          &\mcc_2(t_{2,j_{2,1}};\zeta)&\mcc_2(t_{2,j_{2,1}}^\epsilon;\zeta)&\cdots
          &\mcc_2(t_{2,j_{2,k_2}};\zeta)&\mcc_2(t_{2,j_{2,k_2}}^\epsilon;\zeta)
        \end{smallmatrix}\\[6pt]
        &
        \begin{smallmatrix}
          % {@{\hsep}c@{\hsep}c@{\hsep}c@{\hsep}c@{\hsep}c@{\hsep}c@{\hsep}}
          % c@{\hsep}c@{\hsep}}
          \cdots&1&1&\cdots&1&1\\[3pt]
          \cdots&\mcc_{r-1}(t_{r-1,j_{r-1,1}};\zeta)&\mcc_{r-1}(t_{r-1,j_{r-1,1}}^\epsilon;\zeta)&\cdots
          &\mcc_{r-1}(t_{r-1,j_{r-1,k_{r-1}}};\zeta)&\mcc_{r-1}(t_{r-1,j_{r-1,k_{r-1}}}^\epsilon;\zeta)
        \end{smallmatrix}\\[6pt]
        &
        \begin{smallmatrix}
          % {@{\hsep}c@{\hsep}c@{\hsep}c@{\hsep}c@{\hsep}c@{\hsep}c@{\hsep}}
          &1&1&\cdots&1&1\\[3pt]
          &\mcc_r(t_{r,j_{r,1}};\zeta)&\mcc_r(t_{r,j_{r,1}}^\epsilon;\zeta)
          &\cdots&\mcc_r(t_{r,j_{r,k_r}};\zeta)&\mcc_r(t_{r,j_{r,k_r}}^\epsilon;\zeta)\\
        \end{smallmatrix}\\[6pt]
        &
        \begin{smallmatrix}
          % {@{\hsep}c@{\hsep}c@{\hsep}c@{\hsep}c@{\hsep}c@{\hsep}c@{\hsep}}
          &1&1&\cdots&1\\[3pt]
          &\mcc_r(M;\zeta)&\mcc_r(2M;\zeta)&\cdots&\mcc_r((\dup-2k)M;\zeta)
        \end{smallmatrix}
      \bigg|,
      \end{aligned}
  \end{equation}}%
  where $R=\frac{r(r-1)}{2}$. As for the determinant $H_U(\mb{x})$,
  we can alternatively describe $F_V(\mb{x};\zeta)$ as follows:
  \begin{enumerate}
  \item
    The first column of $F_V(\mb{x};\zeta)$ is $\binom{1}{\mb{x}}$.
  \item
    For $i$ ranging from $1$ to $r$, and for $\lambda$ ranging from $1$ to
    $k_i$, the next $k_i$ \emph{pairs of columns} of
    $F_V(\mb{x})$ are $\binom{1}{\mcc_i(t_{i,j_{i,\lambda}};\zeta)}$ and
    $\binom{1}{\mcc_i(t_{i,j_{i,\lambda}}^\epsilon;\zeta)}$.
  \item
    For $\lambda$ ranging from $1$ to $\dup-2k$,
    the last $\dup-2k$ columns of $F_V(\mb{x};\zeta)$ are
    $\binom{1}{\mcc_r(\lambda M;\zeta)}$. Notice that if $k=\lexp{\dup}$ and
    $\dup$ is even, this last category of columns of $F_V(\mb{x};\zeta)$ does
    not exist.
  \end{enumerate}

  The equation $F_V(\mb{x};\zeta)=0$ is the equation of a
  hyperplane in $\reals^{\dup}$ that passes through the points in
  $V$. We are going to show that for all vertices
  $\mbv\in\VV\sm{}V$, we have $F_V(\mbv;\zeta)>0$ for
  sufficiently small $\zeta$.

  Indeed, choose some $\mbv\in{}\VV\sm{}V$, and suppose that
  $\mbv\in{}\VV_i\sm{}V$. Then $\mbv$ is of the form
  $\mbv=\mcc_i(t_{i,\lambda};\zeta)$, $\zeta>0$, for some
  $\lambda\nin\{j_{i,1},j_{i,2},\ldots,j_{i,k_i}\}$.
  Let $\mbu^\star=\mcc_i(t_{i,\lambda})=\mcc_i(t_{i,\lambda};0)$.
  In more geometric terms, we define $\mbu^\star$ to be the projection of
  $\mbv$ on the $d$-flat $\reals^{r-1}\times{}F_i$ of
  $\reals^{r-1}\times\reals^{d}$, or, equivalently, $\mbu^\star$ is
  the (unperturbed) vertex in $\UU^\star\sm{}U^\star$ that corresponds to
  $\mbv$, where $U^\star$ stands for the set of (unperturbed)
  vertices of $\UU^\star$ that correspond to the vertices in $V$.
  Observe that $F_V(\mbv;\zeta)$ is a polynomial function in
  $\zeta$, and thus it is continuous with respect to $\zeta$ for any
  $\zeta\in\reals$. This implies that
  \begin{equation}\label{equ:limFv}
    \lim_{\zeta\to{}0^+}F_V(\mbv;\zeta)=F_{U^\star}(\mbu^\star;0)
    =H_{U^\star}(\mbu^\star),
  \end{equation}
  where we used the fact that $\lim_{\zeta\to{}0^+}\mbv=\mbu^\star$,
  and observed that $F_{U^\star}(\mbu^\star;0)=H_{U^\star}(\mbu^\star)$, 
  where $H_{U^\star}(\mb{x})$ is the determinant in relation
  \eqref{equ:H-def} in the proof of Lemma \ref{lem:Qstar}, for
  $\tau\sub\tau^\star$.
  Since $H_{U^\star}(\mbu^\star)>0$ (recall that we
  have chosen $\tau$ to be equal to $\tau^\star$), we conclude, from
  \eqref{equ:limFv}, that there exists some $\zeta_0>0$ that
  depends on $\mbv$ and $V$, such that for all
  $\zeta\in(0,\zeta_0)$, $F_V(\mbv;\zeta)>0$.

  Since for each $k$ with $r\le{}k\le{}\lexp{\dup}$, the number
  of \spansubs $V$ of $\VV$ of size $k$ is $\ub{k}$,
  while for each such subset $V$ we need to consider the
  $(\sum_{i=1}^{r}n_i-k)$ vertices in $\VV\sm{}V$, it suffices to
  consider a positive value $\zeta^\lozenge$ for $\zeta$ that is small
  enough, so that all
  \begin{equation*}
    \sum_{k=r}^{\lexp{\dup}}\Big(\sum_{i=1}^{r}n_i-k\Big)\,\ub{k}
  \end{equation*}
  possible determinants $F_V(\mbv;\zeta)$ are strictly positive. 
  Hence, for $\zeta\sub\zeta^\lozenge$, we have that \emph{for each}
  \spansub $V$ of $\VV$ 
  the equation $F_{V}(\mb{x};\zeta^\lozenge)=0$, $\mb{x}\in\reals^{\dup}$,
  is the equation of a supporting hyperplane for $\pP$ passing
  through the vertices of $V$, and those only. In other words,
  every \spansub $V$ of $\VV$, where
  $|V|=k$ and $r\le{}k\le{}\lexp{\dup}$,
  defines a $(k-1)$-face of $\pP$, which means that 
  $f_{k-1}(\wW_{\pP})=\ub{k}$, for all $r\le{}k\le{}\lexp{\dup}$.
\end{proof}

From Lemma \ref{lem:Pstar}, in conjunction with \eqref{equ:fkW},
we immediately arrive at the following theorem, which
states the main result of this paper.

\begin{theorem}\label{thm:lb}
  Let $d\ge{}3$ and $2\le{}r\le{}d-1$.
  There exist $r$ neighborly $d$-polytopes $\range{P}$ in $\reals^d$,
  with $\range{n}$ vertices, respectively, such that, for all
  $0\le{}k\le{}\lexp{d+r-1}-r$:
  \begin{equation*}
    f_{k}(\MS)=\ub{k+r}=
    \sum_{\substack{1\le{}s_i\le{}n_i\\s_1+\ldots+s_r=k+r}}
    \prod_{i=1}^r\binom{n_i}{s_i}.
  \end{equation*}
\end{theorem}

%As a final remark, notice that Theorem \ref{thm:lb} contains, as a
%special case, Fukuda and Weibel's tight bound on the number of
%vertices of the Minkowski sum of $r$ $d$-polytopes, for $d\ge{}3$ and
%$2\le{}r\le{}d-1$ (cf. \cite[Theorem 2]{fw-fmacp-07}).

\section{Discussion}
\label{sec:concl}

In this paper we have considered the problem of evaluating the maximum
number of $k$-faces of the Minkowski sum of $r$ $d$-polytopes in
$\reals^d$. We have shown, with the aid of the Cayley trick for
Minkowski sums, that the trivial upper bound proved by
Fukuda and Weibel \cite[Lemma 8]{fw-fmacp-07} is attainable for any
$d\ge{}3$, $2\le{}r\le{}d-1$ and $0\le{}k\le{}\lexp{d+r-1}-r$, which
is a significant and non-trivial extension of their previous
tightness result (cf. \cite[Theorem 4]{fw-fmacp-07}). A direct
corollary of our lower bounds is that the complexity of the Minkowski
sum of $r$ $n$-vertex $d$-polytopes is in $\Theta(n^{\lexp{d+r-1}})$,
for any fixed $d\ge{}3$ and $2\le{}r\le{}d-1$.
We conjecture that the lower bound construction presented in this
paper, gives, in fact, the maximum possible number of $k$-faces for
the Minkowski sum of $r$ $d$-polytopes for any
$d\ge{}3$, $2\le{}r\le{}d-1$, and for all $0\le{}k\le{}d-1$. Our
conjecture has been positively asserted for the case of two
$d$-polytopes (cf. \cite{kt-mnfms-11b,kt-mnfms-12}).

Given the results in this paper, as well as the tight bounds in
\cite{kt-mnfms-11b,kt-mnfms-12} and \cite{w-mfmsl-11},
the obvious remaining open problem is to devise a tight expression for
the maximum number of $k$-faces of the Minkowski sum of $r$
$d$-polytopes for: (i) $d\ge{}4$, $3\le{}r\le{}d-1$ and
$\lexp{d+r-1}-r<k\le{}d-1$, and (ii) for $r\ge{}d\ge{}3$ and
$1\le{}k\le{}d-1$.
Another relevant open problem is to express the maximum number of
$k$-faces of the Minkowski sum of $r$ $d$-polytopes as a function of
the number of facets of the polytopes. Results in this direction are
known for 2- and 3-polytopes only (cf. \cite{w-mspcc-07},
\cite{fhw-emcms-09}).
%We would like to derive such expressions for any $d\ge{}4$ and any
%number of summands.

\section*{Acknowledgements}
The authors would like to thank Panagiotis D. Kaklis and Christos
Konaxis for discussions regarding earlier versions of this paper.
The work in this paper has been partially supported by the
FP7-REGPOT-2009-1 project ``Ar\-chi\-me\-des Center for Modeling, Analysis
and Computation''.

\phantomsection
\addcontentsline{toc}{section}{References}

\bibliographystyle{plain}
\bibliography{minksum}

\begin{thebibliography}{10}

\bibitem{bkos-cgaa-00}
Mark de~Berg, Marc van Kreveld, Mark Overmars, and Otfried Schwarzkopf.
\newblock {\em Computational Geometry: Algorithms and Applications}.
\newblock Springer-Verlag, Berlin, Germany, 2nd edition, 2000.

\bibitem{fhw-emcms-09}
Efi Fogel, Dan Halperin, and Christophe Weibel.
\newblock On the exact maximum complexity of {M}inkowski sums of polytopes.
\newblock {\em Discrete Comput. Geom.}, 42:654--669, 2009.

\bibitem{f-mscaa-08}
Efraim Fogel.
\newblock {\em Minkowski Sum Construction and other Applications of
  Arrangements of Geodesic Arcs on the Sphere}.
\newblock PhD thesis, Tel-Aviv University, October 2008.

\bibitem{fw-fmacp-07}
Komei Fukuda and Christophe Weibel.
\newblock $f$-vectors of {M}inkowski additions of convex polytopes.
\newblock {\em Discrete Comput. Geom.}, 37(4):503--516, 2007.

\bibitem{g-tm-60}
F.~R. Gantmacher.
\newblock {\em The Theory of Matrices}, volume~I.
\newblock Chelsea Publishing Co., New York, 1960.

\bibitem{g-atm-05}
F.~R. Gantmacher.
\newblock {\em Applications of the Theory of Matrices}.
\newblock Dover, Mineola, New York, 2005.

\bibitem{gs-mapcc-93}
Peter Gritzmann and Bernd Sturmfels.
\newblock Minkowski addition of polytopes: Computational complexity and
  applications to {G}r\"obner bases.
\newblock {\em SIAM J. Disc. Math.}, 6(2):246--269, May 1993.

\bibitem{hk-la-71}
Kenneth Hoffman and Ray Kunze.
\newblock {\em Linear Algebra}.
\newblock Prentice Hall, 2nd edition, 1971.

\bibitem{hrs-ctlsb-00}
Birkett Huber, J\"org Rambau, and Francisco Santos.
\newblock The {C}aylay {T}rick, lifting subdivisions and the {B}ohne-{D}ress
  theorem on zonotopal tilings.
\newblock {\em J. Eur. Math. Soc.}, 2(2):179--198, June 2000.

\bibitem{kt-chsch-11b}
Menelaos~I. Karavelas and Eleni Tzanaki.
\newblock Convex hulls of spheres and convex hulls of convex polytopes lying on
  parallel hyperplanes.
\newblock In {\em Proc. 27th Annu. ACM Sympos. Comput. Geom. (SCG'11)}, pages
  397--406, Paris, France, June 13--15, 2011.

\bibitem{kt-mnfms-11b}
Menelaos~I. Karavelas and Eleni Tzanaki.
\newblock The maximum number of faces of the {M}inkowski sum of two convex
  polytopes, October 2011.
\newblock {\tt arXiv:1106.6254v2 [cs.CG]}.

\bibitem{kt-mnfms-12}
Menelaos~I. Karavelas and Eleni Tzanaki.
\newblock The maximum number of faces of the {M}inkowski sum of two convex
  polytopes.
\newblock In {\em Proceedings of the 23rd ACM-SIAM Symposium on Discrete
  Algorithms (SODA'12)}, Kyoto, Japan, January 17--19, 2012.
\newblock To appear.

\bibitem{s-tovnm-09}
Raman Sanyal.
\newblock Topological obstructions for vertex numbers of {M}inkowski sums.
\newblock {\em J. Comb. Theory, Ser. A}, 116(1):168--179, 2009.

\bibitem{s-amp-04}
Micha Sharir.
\newblock Algorithmic motion planning.
\newblock In J.~E. Goodman and J.~O'Rourke, editors, {\em Handbook of Discrete
  and Computational Geometry}, chapter~47, pages 1037--1064. Chapman \&
  Hall/CRC, London, 2nd edition, 2004.

\bibitem{w-mspcc-07}
Christophe Weibel.
\newblock {\em Minkowski Sums of Polytopes: Combinatorics and Computation}.
\newblock PhD thesis, \'Ecole Polytechnique F\'ed\'erale de Lausanne, 2007.

\bibitem{w-mfmsl-11}
Christophe Weibel.
\newblock Maximal f-vectors of {M}inkowski sums of large numbers of polytope.
\newblock {\em Discrete Comput. Geom.}, November 2011.
\newblock Online first.

\end{thebibliography}

\phantomsection
\addcontentsline{toc}{section}{Appendix}

%\clearpage

\appendix

%%%%%%%%%%%%%%%%%%%%%%%%%%%%%%%%%%%%%%%%%%%%%%%%%%%%%%%%%%%%%%%%%%%%%
%%%%%%%%%%%%%%%%%%%%%%%%%%%%%%%%%%%%%%%%%%%%%%%%%%%%%%%%%%%%%%%%%%%%%

\section{Proof of Lemma 2}
\label{app:signdet}

We start by introducing what is known as
\emph{Laplace's Expansion Theorem} for determinants
(see \cite{g-tm-60,hk-la-71} for details and proofs).
Consider a $n\times{}n$ matrix $A$.
Let $\mb{r}=(r_1,r_2,\ldots,r_k)$, be a vector of $k$ row indices for
$A$, where $1\le{}k<n$ and $1\le{}r_1<r_2<\ldots<r_k\le{}n$.
Let $\mb{c}=(c_1,c_2,\ldots,c_k)$ be a vector of $k$ column indices for
$A$, where $1\le{}k<n$ and $1\le{}c_1<c_2<\ldots<c_k\le{}n$. We
denote by $S(A;\mb{r},\mb{c})$ the $k\times{}k$ submatrix of $A$
constructed by keeping the entries of $A$ that belong to a row in
$\mb{r}$ and a column in $\mb{c}$.
The \emph{complementary submatrix} for $S(A;\mb{r},\mb{c})$, denoted
by $\compl(A;\mb{r},\mb{c})$, is the $(n-k)\times(n-k)$ submatrix of
$A$ constructed by removing the rows and columns of $A$ in $\mb{r}$
and $\mb{c}$, respectively. Then, the determinant of $A$ can be
computed by expanding in terms of the $k$ columns of $A$ in
$\mb{c}$ according to the following theorem.

\begin{theorem}[\textbf{Laplace's Expansion Theorem}]\label{thm:LET}
  Let $A$ be a $n\times{}n$ matrix. Let
  $\mb{c}=(c_1,c_2,\ldots,c_k)$ be a vector of $k$ column indices for
  $A$, where $1\le{}k<n$ and $1\le{}c_1<c_2<\ldots<c_k\le{}n$. Then:
  \begin{equation}
    \det(A)=\sum_{\mb{r}}(-1)^{|\mb{r}|+|\mb{c}|}
    \det(S(A;\mb{r},\mb{c}))\det(\compl(A;\mb{r},\mb{c})),
  \end{equation}
  where $|\mb{r}|=r_1+r_2+\ldots+r_k$,
  $|\mb{c}|=c_1+c_2+\ldots+c_k$, and the summation is taken over all
  row vectors $\mb{r}=(r_1,r_2,\ldots,r_k)$ of $k$ row indices for
  $A$, where $1\le{}r_1<r_2<\ldots<r_k\le{}n$.
%  $k$-tuples $\mb{r}=(r_1,r_2,\ldots,r_k)$ for which
%  $1\le{}r_1<r_2<\ldots<r_k\le{}n$.
\end{theorem}

Given a vector of $n\ge{}2$ real numbers
$\mb{x}=(x_1,x_2,\ldots,x_n)$, the \emph{Vandermonde determinant}
$\VD{x}$ of $\mb{x}$ is the $n\times{}n$ determinant
\begin{equation}\label{equ:VD}
  \VD{x}=\left|
    \begin{array}{ccccc}
      1&1&\cdots&1\\
      x_1&x_2&\cdots&x_n\\
      x_1^2&x_2^2&\cdots&x_n^2\\
      \vdots&\vdots&&\vdots\\
      x_1^{n-1}&x_2^{n-1}&\cdots&x_n^{n-1}\\
    \end{array}
  \right|=\prod_{1\le{}i<j\le{}n}(x_j-x_i)
\end{equation}
From the above expression, it is readily seen that if the elements of
$\mb{x}$ are in strictly increasing order, then $\VD{x}>0$. A
generalization of the Vandermonde determinant is the generalized
Vandermonde determinant: if, in addition to $\mb{x}$, we specify a vector
of exponents $\mb{\mu}=(\mu_1,\mu_2,\ldots,\mu_n)$, where
we require that $0\le{}\mu_1<\mu_2<\ldots<\mu_n$, we can
define the \emph{generalized Vandermonde determinant}
$\GVD(\mb{x};\mb{\mu})$ as the $n\times{}n$ determinant:
\begin{equation*}
  \GVD(\mb{x};\mb{\mu})=\left|
    \begin{array}{ccccc}
      x_1^{\mu_1}&x_2^{\mu_1}&\cdots&x_n^{\mu_1}\\
      x_1^{\mu_2}&x_2^{\mu_2}&\cdots&x_n^{\mu_2}\\
      x_1^{\mu_3}&x_2^{\mu_3}&\cdots&x_n^{\mu_3}\\
      \vdots&\vdots&&\vdots\\
      x_1^{\mu_n}&x_2^{\mu_n}&\cdots&x_n^{\mu_n}\\
    \end{array}
  \right|.
\end{equation*}
It is a well-known fact (\eg see \cite{g-atm-05}) that, if the
elements of $\mb{x}$ are in strictly increasing order, then
$\GVD(\mb{x};\mb{\mu})>0$.

We now restate Lemma \ref{lem:sign-generic-det} and prove it.

\restorecounter{theorem}

\techlemma{\label{lem:sign-generic-det1}}

\begin{proof}
  We denote by $\md$ the matrix corresponding to the determinant
  $\mddet$. We are going to apply Laplace's Expansion
  Theorem to evaluate $\mddet$ in terms of the first $\kappa_1$
  columns of $\md$, then the next $\kappa_2$ columns of $\md$, then the
  next $\kappa_3$ columns of $\md$, and so on. Let $K_i=\sum_{j=1}^{i}\kappa_j$,
  $1\le{}i\le{}n$, and let $\mb{c}_i$, $1\le{}i\le{}n$, be the column
  vector corresponding to the columns
  $K_{i-1}+1$ to $K_i$ (by convention $K_0=0$), \ie
  $\mb{c}_i=(K_{i-1}+1,K_{i-1}+2,\ldots,K_i-1,K_i)$, $1\le{}i\le{}r$.
  By applying Laplace's Expansion Theorem we get:
  \begin{equation}\label{equ:DLET2}
    \begin{aligned}
      \mddet&=\sum_{(\mb{r}_1,\mb{r}_2,\ldots,\mb{r}_{n})}
      (-1)^{\sigma(\mb{r},\mb{c})+N}
      \prod_{i=1}^{n}\det(S(\md;\mb{r}_i,\mb{c}_i))
    \end{aligned}
  \end{equation}
  where $\sigma(\mb{r},\mb{c})$ is an expression that depends on
  $\mb{r}$ and $\mb{c}$, while $\det(S(\md;\mb{r}_i,\mb{c}_i))$
  is the $\kappa_i\times{}\kappa_i$
  submatrix of $\md$ formed by taking the elements of $\md$ that
  belong to the $\kappa_i$ columns in $\mb{c}_i$ and the $\kappa_i$ rows in
  $\mb{r}_i$.

  It is easy to verify that the above sum consists of
  $\prod_{i=1}^{n}\binom{K-K_{i-1}}{\kappa_i}$ terms.
  Observe that, among these terms:
  \begin{enumerate}
  \item
    all terms for which $\mb{r}_i$ contains the $j$-th row, where
    $1\le{}j\le{}2n$ and $j\nin\{i,n+i\}$ vanish, since the matrix
    $S(\md;\mb{r}_i,\mb{c}_i)$, consists of at least one row of zeros,
    and,
  \item
    all terms for which $\mb{r}_i$ does not contain the $i$-th or the
    $(n+i)$-th row vanish, since at least one of the matrices
    $S(\md;\mb{r}_j,\mb{c}_j)$, $j\ne{}i$, consists of a row of
    zeros.
  \end{enumerate}
  In other words, the row vector $\mb{r}_i$ has to be of the form
  $\mb{r}_i=(i,n+i,r_{i,3},r_{i,4},\ldots,r_{i,\kappa_i})$, while
  the remaining terms of the expansion are the
  $\prod_{i=1}^{n}\binom{K-2n-K_{i-1}}{\kappa_i-2}$ terms for which
  $\mb{r}_i=(i,n+i,r_{i,3},r_{i,4},\ldots,r_{i,\kappa_i})$, with
  $2n+1\le{}r_{i,3}<r_{i,4}<\ldots<r_{i,\kappa_i}\le{}K$. For any given such
  $\mb{r}_i$, we have that
  $\det(S(\md;\mb{r}_i,\mb{c}_i))$ is the $\kappa_i\times{}\kappa_i$
  generalized Vandermonde determinant
  $\GVD(\tau^{\beta_i}\mb{x}_i;\mb{r}_i-\mb{\alpha}_i)$, where
  $\mb{x}_i=(x_{i,1},x_{i,2},\ldots,x_{i,\kappa_i})$ and
  $\mb{\alpha}_i=(i,n+i-1,2n-1,2n-1,\ldots,2n-1)
  =i\me_1+(n+i-1)\me_2+(2n-1)\sum_{j=3}^{\kappa_i}\me_j$;
  here $\me_j$ stands for the vector whose elements vanish,
  except for the $j$-th element, which is equal to 1.

  We can, thus, simplify the expansion in \eqref{equ:DLET2} to get:
  \begin{equation}\label{equ:DLET2-simple}
    \mddet=
    \sum_{\substack{(\mb{r_1},\mb{r}_2,\ldots,\mb{r}_n)\\
        \mb{r}_\nu=(\nu,n+\nu,r_{\nu,3},\ldots,r_{\nu,\kappa_\nu})\\
        2n+1\le{}r_{\nu,3}<r_{\nu,4}<\ldots<r_{\nu,\kappa_\nu}\le{}K}}
    (-1)^{\sigma(\mb{r},\mb{c})+N}
    \prod_{i=1}^n\GVD(\tau^{\beta_i}\mb{x}_i;\mb{r}_i-\mb{\alpha}_i)
  \end{equation}
  Since
  \begin{equation*}
    \GVD(\tau^{\beta_i}\mb{x}_i;\mb{r}_i-\mb{\alpha}_i)
    =\tau^{\beta_i|\mb{r}_i|-\beta_i|\mb{\alpha}_i|}
    \,\GVD(\mb{x}_i;\mb{r}_i-\mb{\alpha}_i),
  \end{equation*}
  the expansion for $\mddet$ can be further rewritten as:
  \begin{equation}\label{equ:DLET2-simple2}
    \mddet=
    %\!\!\!\!\!\!\!\!\!\!\!\!
    \sum_{\substack{(\mb{r_1},\mb{r}_2,\ldots,\mb{r}_n)\\
        \mb{r}_\nu=(\nu,n+\nu,r_{\nu,3},\ldots,r_{\nu,\kappa_\nu})\\
        2n+1\le{}r_{\nu,3}<r_{\nu,4}<\ldots<r_{\nu,\kappa_\nu}\le{}K}}
    %\!\!\!\!\!\!\!\!\!\!\!\!
    (-1)^{\sigma(\mb{r},\mb{c})+N}
    \tau^{\sum_{i=1}^{n}\beta_i(|\mb{r}_i|-|\mb{\alpha}_i|)}\prod_{i=1}^n
    \,\GVD(\mb{x}_i;\mb{r}_i-\mb{\alpha}_i)
  \end{equation}

  Let $\mb{\rho}=(\mb{\rho}_1,\mb{\rho}_2,\ldots,\mb{\rho}_n)$
  be the row vector for which
  $\mb{\rho}_i=(i,n+i,2n+K_{i-1}-2(i-1)+1,2n+K_{i-1}-2(i-1)+2,\ldots,2n+K_i-2i)
  =(i,n+i,2(n-i)+K_{i-1}+3,2(n-i)+K_{i-1}+4,\ldots,2(n-i)+K_i)$,
  $1\le{}i\le{}n$. We claim that the minimum exponent for $\tau$ in
  \eqref{equ:DLET2-simple2} is attained when
  $\mb{r}_i\equiv\mb{\rho}_i$, for all $1\le{}i\le{}n$.
  Suppose, on the contrary, that the minimum exponent for $\tau$ is
  attained for the row vector
  $\mb{r}=(\mb{r}_1,\mb{r}_2,\ldots,\mb{r}_n)\ne\mb{\rho}$.
  Let $i_1$, $1\le{}i_1\le{}n$, be the minimum index $i$ for which
  $\mb{r}_i\ne\mb{\rho}_i$, and let $r_{i_1,\ell_1}$,
  $3\le\ell_1\le{}\kappa_{i_1}$, be the first element of $\mb{r}_{i_1}$ that
  differs from the corresponding element of $\mb{\rho}_{i_1}$, \ie
  $r_{i_1,j}=\rho_{i_1,j}$, $1\le{}j\le{}\ell_1-1$, and 
  $r_{i_1,\ell_1}\ne\rho_{i_1,\ell_1}$.
  Since both $\mb{r}$ and $\mb{\rho}$ are row vectors, there exists
  another index $i_2$ with $1\le{}i_1<i_2\le{}n$, such that
  $\mb{r}_{i_2}$ contains $\rho_{i_1,\ell_1}$, and let $\ell_2$,
  $3\le{}\ell_2\le{}\kappa_{i_2}$, be the index of $\rho_{i_1,\ell_1}$ in
  $\mb{r}_{i_2}$, i.e., $r_{i_2,\ell_2}=\rho_{i_1,\ell_1}$. Define the
  row vector
  $\mb{\bar{r}}=(\mb{\bar{r}}_1,\mb{\bar{r}}_2,\ldots,\mb{\bar{r}}_n)$
  as follows.
  Let $\mb{\bar{r}}_i=\mb{r}_i$, for all $i\ne{}i_1,i_2$,
  $\mb{\bar{r}}_{i_1}$ has the same elements as $\mb{r}_{i_1}$, except
  for its $\ell_1$-th element, which is replaced by
  $\rho_{i_1,\ell_1}$ (i.e.,
  $\bar{r}_{i_1,\ell_1}=\rho_{i_1,\ell_1}<r_{i_1,\ell_1}$), whereas
  $\mb{\bar{r}}_{i_2}$ has the same elements as $\mb{r}_{i_2}$, except
  for its $\ell_2$-th element, which is replaced by $r_{i_1,\ell_1}$
  (i.e.,
  $\bar{r}_{i_2,\ell_2}=r_{i_1,\ell_1}>\rho_{i_1,\ell_1}=r_{i_2,\ell_2}$). This
  implies that $|\mb{\bar{r}}_i|=|\mb{r}_i|$, for all $i\ne{}i_1,i_2$,
  whereas
  \begin{align*}
    |\mb{\bar{r}}_{i_1}|&=|\mb{r}_{i_1}|-r_{i_1,\ell_1}+\rho_{i_1,\ell_1},\\
    |\mb{\bar{r}}_{i_2}|&=|\mb{r}_{i_2}|-r_{i_2,\ell_2}+r_{i_1,\ell_1}
    =|\mb{r}_{i_2}|-\rho_{i_1,\ell_1}+r_{i_1,\ell_1}.
  \end{align*}
  Hence we get:
  {\allowdisplaybreaks
  \begin{align*}
    \sum_{i=1}^n\beta_i|\mb{\bar{r}}_i|
    &=\sum_{\substack{i=1\\i\ne{}i_1,i_2}}^n\beta_i|\mb{\bar{r}}_i|
    +\beta_{i_1}|\mb{\bar{r}}_{i_1}|+\beta_{i_2}|\mb{\bar{r}}_{i_2}|\\
    &=\sum_{\substack{i=1\\i\ne{}i_1,i_2}}^n\beta_i|\mb{r}_i|
    +\beta_{i_1}(|\mb{r}_{i_1}|-r_{i_1,\ell_1}+\rho_{i_1,\ell_1})
    +\beta_{i_2}(|\mb{r}_{i_2}|-\rho_{i_1,\ell_1}+r_{i_1,\ell_1})\\
    &=\sum_{i=1}^n\beta_i|\mb{r}_i|
    +(\beta_{i_1}-\beta_{i_2})(\rho_{i_1,\ell_1}-r_{i_1,\ell_1})\\
    &<\sum_{i=1}^n\beta_i|\mb{r}_i|,
  \end{align*}}%
  where we used that $\rho_{i_1,\ell_1}<r_{i_1,\ell_1}$ and
  $\beta_{i_1}>\beta_{i_2}$ (since $i_1>i_2$). This, however,
  contradicts the minimality property of $\mb{r}$, which means that
  our assumption that $\mb{r}\ne\mb{\rho}$ is false.

  For $\mb{r}\equiv\mb{\rho}$, it is easy to verify that
  \begin{equation}\label{equ:sigma-expr}
    \sigma(\mb{\rho},\mb{c})=\sum_{i=1}^{n-1}\sum_{j=1}^{\kappa_i}{j}
    +\sum_{i=1}^{n-1}\left[1+(n+2-i)+\sum_{j=1}^{\kappa_i-2}[2(n+1-i)+j]\right].
  \end{equation}
  To see why this expression holds, consider expanding $\mddet$ along
  the columns in $\mb{c}_1$ and the rows in $\mb{\rho}_1$. This
  contributes a term of $\sum_{j=1}^{\kappa_1}{j}$ to
  $\sigma(\mb{\rho},\mb{c})$, corresponding to $|\mb{c}_1|$, and a
  term $1+(n+1)+\sum_{j=1}^{\kappa_1-2}(2n+j)$ corresponding to
  $|\mb{\rho}_1|$. The remaining complementary
  $(K-\kappa_1)\times(K-\kappa_1)$ submatrix is then expanded along
  the column vector corresponding to its first $\kappa_2$ columns, and
  the row vector corresponding to its first row, its $n$-th row, as
  well as rows $2(n-1)+j$, with $1\le{}j\le{}\kappa_2-2$. This
  contributes a term $\sum_{j=1}^{\kappa_2}j$ to $\sigma(\mb{\rho},\mb{c})$
  corresponding to its first $\kappa_2$ columns, and a term of
  $1+n+\sum_{j=1}^{\kappa_2-2}(2(n-1)+j)$ corresponding to its
  rows. At the $i$-th step of this procedure, which is performed
  $(n-1)$ times, the remaining determinant corresponds to a
  $(K-K_{i-1})\times(K-K_{i-1})$ submatrix of $\md$, which is expanded
  with respect to its first $\kappa_i$ columns, and with respect to its
  first and $(n+2-i)$-th row, as well as its rows $2(n+1-i)+j$, where
  $1\le{}j\le{}\kappa_i-2$. Hence the contribution to
  $\sigma(\mb{\rho},\mb{c})$ from the columns in $\mb{c}_i$ is
  $\sum_{j=1}^{\kappa_i}j$, while the contribution from the rows in
  $\mb{\rho}_i$ is $1+(n+2-i)+\sum_{j=1}^{\kappa_i-2}[2(n+1-i)+j]$.
  Expanding $\sigma(\mb{\rho},\mb{c})$ from \eqref{equ:sigma-expr},
  and gathering terms appropriately, we get:
  {\allowdisplaybreaks
  \begin{align*}
    \sigma(\mb{\rho},\mb{c})&=\sum_{i=1}^{n-1}\sum_{j=1}^{\kappa_i}{j}
    +\sum_{i=1}^{n-1}\left[1+(n+2-i)+\sum_{j=1}^{\kappa_i-2}[2(n+1-i)+j]\right]\\
    &=\sum_{i=1}^{n-1}\sum_{j=1}^{\kappa_i}{j}
    +\sum_{i=1}^{n-1}(n+3-i)
    +2\sum_{i=1}^{n-1}\sum_{j=1}^{\kappa_i-2}(n+1-i)
    +\sum_{i=1}^{n-1}\sum_{j=1}^{\kappa_i-2}j\\
    &=\sum_{i=1}^{n-1}\sum_{j=1}^{\kappa_i}{j}
    +3(n-1)+\sum_{i=1}^{n-1}(n-i)
    +2\sum_{i=1}^{n-1}\sum_{j=1}^{\kappa_i-2}(n+1-i)
    +\sum_{i=1}^{n-1}\sum_{j=1}^{\kappa_i}j-\sum_{i=1}^{n-1}[(\kappa_i-1)+\kappa_i]\\
    &=2\sum_{i=1}^{n-1}\sum_{j=1}^{\kappa_i}{j}
    +3(n-1)+\sum_{i=1}^{n-1}i
    +2\sum_{i=1}^{n-1}\sum_{j=1}^{\kappa_i-2}(n+1-i)
    -2\sum_{i=1}^{n-1}\kappa_i+\sum_{i=1}^{n-1}1\\
    &=2\sum_{i=1}^{n-1}\sum_{j=1}^{\kappa_i}{j}
    +4(n-1)+\frac{n(n-1)}{2}
    +2\sum_{i=1}^{n-1}\sum_{j=1}^{\kappa_i-2}(n+1-i)
    -2\sum_{i=1}^{n-1}\kappa_i\\
    &=2\left[\sum_{i=1}^{n-1}\sum_{j=1}^{\kappa_i}{j}
    +2(n-1)+\sum_{i=1}^{n-1}\sum_{j=1}^{\kappa_i-2}(n+1-i)
    -\sum_{i=1}^{n-1}\kappa_i\right]+\frac{n(n-1)}{2}.
  \end{align*}}%
  This means that
  $(-1)^{\sigma(\mb{\rho},\mb{c})+N}=(-1)^{n(n-1)}=1$, since $n(n-1)$
  is always even. Therefore, relation \eqref{equ:DLET2-simple2} gives:
  \begin{equation}\label{equ:DLET2-asymp}
    \mddet=
    \tau^{\sum_{i=1}^{n}\beta_i(|\mb{\rho}_i|-|\mb{\alpha}_i|)}
    \prod_{i=1}^n\,\GVD(\mb{x}_i;\mb{\rho}_i-\mb{\alpha}_i)
    +O(\tau^{\sum_{i=1}^{n}\beta_i(|\mb{\rho}_i|-|\mb{\alpha}_i|)+1}).\\
  \end{equation}
  From relation \eqref{equ:DLET2-asymp} we immediately deduce that:
  \begin{equation}
    \lim_{\tau\to{}0^+}\frac{\mddet}{\tau^{\sum_{i=1}^{n}\beta_i(|\mb{\rho}_i|-|\mb{\alpha}_i|)}}
    =\prod_{i=1}^n\,\GVD(\mb{x}_i;\mb{\rho}_i-\mb{\alpha}_i)
  \end{equation}
  which establishes the claim of the lemma, since
  $\GVD(\mb{x}_i;\mb{\rho}_i-\mb{\alpha}_i)$ is strictly positive, for
  all $1\le{}i\le{}n$.
\end{proof}

\end{document}